\documentclass[11pt]{article}
\usepackage{graphicx}
\usepackage{algorithm}
\usepackage[noend]{algorithmic}
\usepackage{hyperref}

\usepackage{amsmath, amsthm, amssymb}

\usepackage[latin1]{inputenc}

\newcommand{\remove}[1]{}
\newcommand{\comments}[1]{}

\newtheorem{definition}{Definition}

\newtheorem{corollary}{Corollary}
\newtheorem{lemma}{Lemma}
\newtheorem{theorem}{Theorem}
\newtheorem{proposition}{Proposition}

\begin{document}

\title{Decision Trees for Function Evaluation\\
Simultaneous Optimization of   
Worst and Expected Cost\footnote{A preliminary version of this paper was accepted for presentation at ICML 2014}}

\author{Ferdinando Cicalese\\  University of Salerno, Italy\\
{\tt cicalese@dia.unisa.it} \and 
Eduardo Laber \\PUC-Rio, Brazil\\ {\tt laber@inf.puc-rio.br} \and Aline Medeiros Saettler\\ 
PUC-Rio, Brazil\\ {\tt alinemsaettler@gmail.com}}


\maketitle

\vskip 0.3in

\begin{abstract}
In several applications of automatic diagnosis  and active learning a central 
problem is the evaluation of a discrete function by adaptively querying the values of its variables until 
the values read uniquely determine the value of the function.  
In general, the process of reading the value of a variable 
might involve some cost, computational or even a fee to be paid for the experiment required for obtaining the value.
This cost should be taken into account when deciding the next variable to read. 
The goal is to 
design a strategy for evaluating the function incurring little cost (in the worst case or in expectation 
according to a prior distribution on the possible variables' assignments).

Our algorithm builds a strategy (decision tree) which attains a logarithmic approximation simultaneously for the 
expected and worst cost spent. 
This is best possible under the assumption that ${\cal P} \neq {\cal NP}$.

\end{abstract}

\section{Introduction}

In order to introduce the problem we analyze in the paper, let us start with some motivating examples.

In high frequency trading, an automatic agent 
decides the next action to be performed  as sending or canceling a  buy/sell order,
on the basis of some market variables as well
as private variables (e.g.,  stock price, traded volume,  volatility,
order books distributions as well as complex relations among these variables).
For instance in  \cite{conf/icml/NevmyvakaFK06} the trading strategy is learned in the form
of a discrete function, described as a table, that has to be evaluated whenever a new scenario is faced and 
an action (sell/buy) has to be taken. 
The rows of the table represent the possible scenarios of the market and the columns represent the variables taken into 
account by the agent to distinguish among the different scenarios. For each scenario, there is an associated action.
Every time an action need to be taken, the agent can identify the scenario by computing the  value of each single variable and proceed with the
associated action. However, recomputing all the variable every time might be very expensive. 
By taking into account the structure of the function/table together with 
information on the probability distribution on the scenarios of the market and also the fact that some variables are more
expensive (or time consuming) to calculate than others, the algorithm could limit itself to recalculate only some variables whose values 
determine the action to be taken. Such an approach  can significantly speed up the evaluation of the function. 
Since market conditions change on a millisecond basis, being able to react  very quickly to a new scenario is the key to a profitable strategy.

In a classical Bayesian active learning problem, the task is to select the right hypothesis from a possibly very large set 
${\cal H} = \{h_1,\ldots,h_n\}.$ Each $h \in {\cal H}$ is a mapping from a set ${\cal X}$ called the query/test space 
to the set (of labels) $\{1, \dots, \ell\}.$ It is assumed that the functions in ${\cal H}$ are unique, i.e., for each pair of them there is 
at least one point in ${\cal X}$ where they differ. There is one function $h^* \in {\cal H}$ which provides the correct labeling of the 
space ${\cal X}$ and the task is to identify it through queries/tests.  A query/test coincides with an element 
$x \in {\cal X}$ and the result is the  value $h^*(x).$ Each test $x$ has an associated cost $c(x)$ that must be paid in order to 
acquire the response $h^*(x),$ 
since the process of  labeling an example may  be expensive either in terms of time or money (e.g. annotating a document).
The goal is to identify the correct hypothesis spending  as little as possible. For instance, 
in  automatic diagnosis, $\cal H$ represents the set of possible diagnoses and $\cal X$ the set of 
symptoms or medical tests, with $h^*$ being the exact diagnosis that has to be achieved by reducing the cost of the examinations.

In \cite{bellala}, a more general variant of the problem was considered where rather than the diagnosis it is important to 
identify the therapy (e.g., for cases of poisoning it is important to quickly understand which antidote to administer rather than  
identifying the exact poisoning). This problem can be modeled by defining a partition $\cal P$ on $\cal H$ with each class of 
$\cal P$  representing the subset of diagnoses which requires the same therapy. The problem is then how to identify the class
of the exact $h^*$ rather than $h^*$ itself.   
This model has also been studied by Golovin et al.\  \cite{golovin} to tackle the problem of erroneous tests' responses 
in Bayesian active learning.

The above examples can all be cast into  the following general problem. \\

\noindent
{\bf The Discrete Function Evaluation Problem} (DFEP). An instance of the problem is defined by  a quintuple  $(S, C, T, {\bf p}, {\bf c}),$ 
where $S = \{s_1, \dots, s_n\}$ is a set of objects, 
$C = \{C_1, \dots, C_m\}$ is a partition of $S$ into $m$ classes, $T$ is a set of tests, ${\bf p}$ is a probability distribution on 
$S,$ and ${\bf c}$ is a cost function assigning to each test $t$ a cost $c(t) \in \mathbb{N^+}.$ 
A test $t \in T$,  when applied to an object  $s \in S$, incurs a cost $c(t)$ and  outputs a number $t(s)$ in the set  $\{1,\ldots,\ell\}$.
It is assumed that the set of tests is complete, in the sense that for any distinct $s_1, s_2 \in S$ there exists a test
$t$ such that $t(s_1) \neq t(s_2).$ 
The goal is to define a testing procedure which uses tests from $T$ and minimizes the testing cost 
(in expectation and/or in the worst case)
for identifying the class of an unknown object $s^*$ chosen according to the distribution ${\bf p}.$

The DFEP can be rephrased in terms of minimizing the cost of evaluating a discrete function that maps points (corresponding to objects) from some finite subset of $\{1,\ldots,\ell\}^{|T|}$ into
values (corresponding to classes), where  an object $s\in S$ corresponds to the point $(t_1(s),\ldots,t_{|T|}(s))$ 
 obtained by applying each test of  $T$ to $s$.  
This perspective motivates the name we chose for the problem. However, for the sake of uniformity with more recent work \cite{golovin,bellala} we employ the definition of the problem in terms of objects/tests/classes. \\

\noindent
{\bf Decision Tree Optimization.} Any testing procedure can be represented by a \emph{decision tree}, which is a tree
where every internal node is associated with a test and every leaf is associated with a set
of objects that belong to the same class.
More formally, a decision tree $D$ for $(S,C,T,\mathbf{p}, {\bf c})$  is a leaf
associated with class $i$ 
if every object of $S$ belongs to the same class $i$. Otherwise, the root $r$  of $D$ is associated with some test $t \in  T$ and  
the children of $r$ are  decision trees for the sets $\{S_{t}^1,...,S_{t}^{\ell}\}$,
where $S_{t}^i$, for $i=1,\ldots, \ell$, is the subset of $S$ that outputs $i$ for test $t$.
 
Given a decision tree $D$, rooted at $r$, we can identify the class of an unknown object 
$s^*$ by following a path from  $r$  to a leaf as follows:
first, we ask for the result of the test associated with $r$ when performed on $s^*$; then,  
we follow the branch of $r$  associated with the result of the test  to
reach a child $r_i$ of $r$; next,  we apply the same steps recursively for the decision tree rooted at $r_i$.
The procedure ends when a leaf is reached, which determines the class of $s^*$.

We define $cost(D,s)$ as the sum of the tests' cost on the root-to-leaf path from the root of 
$D$ to the leaf associated with object $s$. 
Then, the \emph{worst testing cost}  and the \emph{expected testing cost} of $D$ are, respectively, defined as

\begin{equation}
cost_W(D) = \max_{s \in S}\{cost(D,s)\} 
\end{equation}

\begin{equation}
 cost_E(D) = \sum_{s \in S} cost(D,s) p(s)
\end{equation}

Figure \ref{fig:decisiontree0} shows an instance of the DFEP and a decision tree for it. The tree has 
 worst testing cost  $1 + 3 + 2 = 6$ and  expected testing cost  $(1\times 0.1)+(6\times  0.2)+(6 \times  0.4)+(4 \times 0.3) = 4.9$.
\begin{figure}
\begin{center}
\includegraphics[width=110mm]{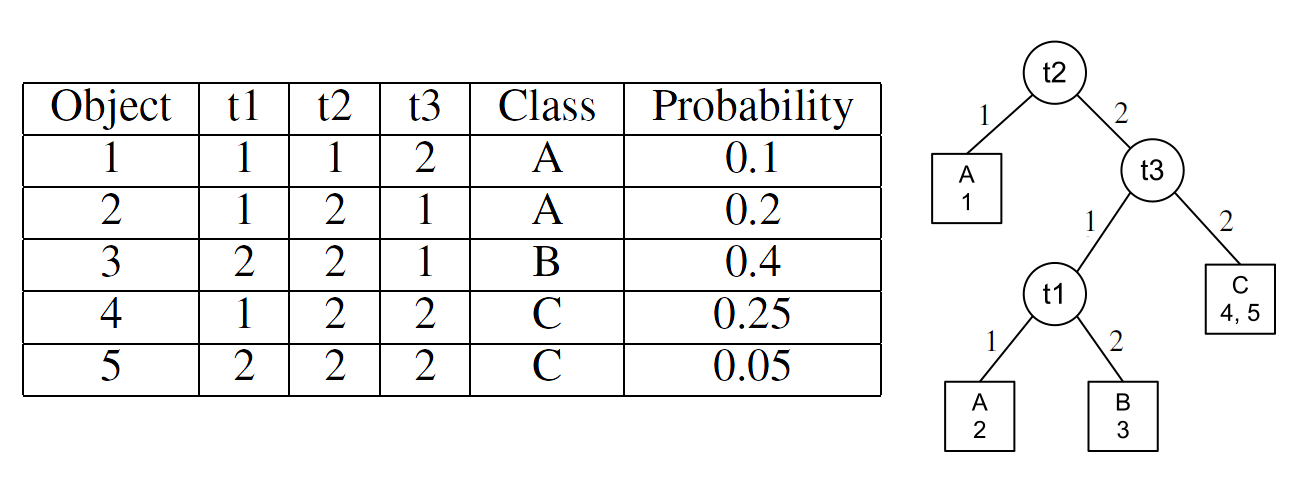}
\label{fig:decisiontree0}
  \caption{Decision tree for 5 objects presented in the table in the left, with $c(t1) = 2$, $c(t2) = 1$ and $c({t3}) = 3$. Letters and 
numbers in the leaves indicate, respectively, classes and objects. 
}
\end{center}
\vskip -0.3in
\end{figure}

\noindent {\bf Our Results.} 
Our main result 
is an algorithm that builds a decision tree whose expected testing cost and worst testing cost are 
at most $O(\log n)$ times the minimum possible expected testing cost and the minimum possible worst testing cost, respectively.
In other words, the decision tree built by our algorithm achieves simultaneously  
the best possible approximation achievable with respect to both the
expected testing cost and the worst testing cost. In fact, for 
the special case where each object defines a distinct class---known as the {\em identification problem}---  
both the minimization of the 
expected testing cost and the minimization of the worst testing cost do not admit a 
sub-logarithmic approximation unless $P=NP,$ as shown in \cite{pandit} and in \cite{laber2}, respectively.
In addition, in Section \ref{sec:inapprox}, we show that the same inapproximability results holds in general  
for the case of exactly $m$ classes for any $m \geq 2.$

It should be noted that in general there are instances
for which the decision tree that minimizes
the expected testing cost has worst testing cost much larger than that achieved by the decision tree with minimum worst testing cost.
Also there are instances where the converse  happens.  
Therefore, it is reasonable to ask whether it is possible to 
construct decision trees that are efficient with respect to both performance criteria.
This might be important in practical applications where only an estimate of  the probability distribution is available which is
not very accurate. Also, in medical applications like the one depicted in \cite{bellala}, very high cost (or equivalently significantly time consuming therapy identification) might have disastrous/deadly consequences.
In such cases, besides being able to minimize the expected testing cost, it is important to guarantee that the worst testing cost also is not large 
(compared with the optimal worst testing cost).

With respect to the minimization of the expected testing cost, our result improves upon the previous $O(\log 1/p_{\min})$ approximation shown in 
\cite{golovin} and  \cite{bellala}, where  $p_{min}$ is the minimum positive probability among the objects in $S$.
From the result in these papers an 
$O(\log n)$ approximation 
could be attained only for the particular case of uniform costs via a technique used in \cite{KosPrzBor99}. 

From a high-level perspective, our method closely follows the one used by Gupta {\em et al.} \cite{gupta} 
for obtaining the $O(\log n)$ approximation for the expected testing cost in  the identification problem.
Both constructions of the decision tree consist of building a path (backbone) that splits the input instance into smaller ones, for which decision trees are recursively constructed and attached as children of the nodes in the path. 
 
A closer look, however, reveals that our algorithm is much simpler than the one presented in \cite{gupta}. First, it is more transparently linked to the structure of the problem, which remained somehow hidden in \cite{gupta} where the result was obtained via an involved mapping from adaptive TSP. Second, our algorithm avoids expensive computational steps as the Sviridenko procedure \cite{Sviridenko} and some non-intuitive/redundant steps that are used to select the tests for the backbone of the tree. In fact, we believe that providing an algorithm that is much simpler to implement and an alternative proof of the result in \cite{gupta} is an additional contribution of this paper.

\medskip

\noindent {\bf State of the art.}
The DFEP has been recently studied under the names of
class equivalence problem \cite{golovin} and group identification problem \cite{bellala} and long before it had been 
described  in the  excellent survey by Moret \cite{More82}.
Both \cite{golovin} and  \cite{bellala} give  $O(\log (1/p_{min}))$ approximation algorithms for the version of the DFEP
where  the expected testing cost has to be minimized and
both the probabilities and the testing costs are non-uniform.
In addition, when the testing costs are uniform both algorithms can be converted into a $O( \log n)$
approximation algorithm via Kosaraju approach \cite{KosPrzBor99}.
The algorithm in \cite{golovin} is more general because it addresses multiway tests rather than binary ones.
For the minimization of the worst testing cost, Moshkov has studied the problem in the general case of multiway tests and non-uniform costs and provided 
an $O(\log n)$-approximation in \cite{Moshkov2}. In the same paper it is also proved that 
no $o(\log n)$-approximation algorithm is possible under standard the complexity assumption
$NP \not \subseteq DTIME(n^{O(\log \log n)}).$ The minimization of the worst testing cost is also investigated in \cite{conf/icml/GuilloryB11} 
under the framework of covering and learning.

\medskip

The particular case of  the DFEP where each object
belongs to a different class---known  as the {\em identification problem}---has been more extensively investigated
\cite{dasgupta1,adler,pandit,pandit2}.
Both the minimization of the worst and the expected testing cost
do not admit a sublogarithmic approximation unless $P=NP$ as proved by \cite{laber2} and
\cite{pandit}. For the expected testing cost, 
in the variant  with multiway tests, non uniform  probabilities and  non uniform testing costs,
an  $O(\log (1/p_{min)})$ approximation is given by Guillory and Blimes in  \cite{guillory}.
Gupta {\em et al.} \cite{gupta}  improved this result to $O(\log n)$ employing new techniques not relying on 
the Generalized Binary Search (GBS)---the basis of all the previous strategies.

An $O(\log n)$ approximation algorithm for the minimization of the worst testing cost for the identification problem has been given by 
Arkin et.\ al.\ \cite{Arkin} for  binary tests and uniform cost and by Hanneke \cite{hanneke}
for case with
mutiway tests and non-uniform testing costs. 

\medskip

In the case of Boolean functions, the DFEP is also known as Stochastic Boolean Function Evaluation (SBFE), where the distribution 
over the possible assignments is a product distribution defined by assuming that variable $x_i$ has a given probability $p(x_i)$ of being one independently 
of the value of the other variables. Another difference with respect to the DFEP as it is presented here, is that in 
Stochastic Boolean Function Evaluation the common assumption is that the complete set of associations between the assignments of the variables
and the value of the function is provided, directly or via a representation of the function, e.g., in terms of its DNF or CNF. The present 
definition of DFEP considers the more general problem where only a sample of the Boolean function is given and from this we want to construct 
a decision tree with minimum expected costs and that exactly fits the sample. 

Results on the exact solution of the SBFE for different classes of Boolean functions 
can be found in the survey paper \cite{Unluyurt}. 
In a  recent paper  Deshpande et al. \cite{Hell}, provide a $3$-approximation algorithm for evaluating Boolean linear threshold formulas and 
an $O(\log k d)$ approximation algorithm for the evaluation of CDNF formulas, where
$k$ and $d$ is  the number of clauses of the input CNF and $d$ is the number of terms of the input DNF. The same result had been previously 
obtained by Kaplan et al. \cite{kaplan} for the case of monotone formulas and uniform distribution (in a slightly different setting).  
Both algorithms of \cite{Hell} are based on reducing the problem to Stochastic Submodular Set Cover introduced by Golovin and Krause \cite{golovin2}
and providing a new algorithm for this latter problem.

Other special cases of the DFEP like the evaluation of AND/OR trees (a.k.a.\ read-once formulas) and the evaluation of Game Trees 
(a central task in the design of game procedures) are discussed in \cite{mmt/Tars83a,DBLP:conf/focs/SaksW86,greiner-etal:2006a}. 
In \cite{CharikarEtAl02a}, Charikar {\em et al.} considered discrete function evaluation from the perspective of 
competitive analysis; results in this alternative setting are also given in \cite{kaplan,DBLP:journals/jacm/CicaleseL11}.

\section{Preliminaries}

Given an instance $I = (S, C, T, {\bf p}, {\bf c})$ of the DFEP, we will denote by $OPT_E(I)$ ($OPT_W(I)$) the 
expected testing cost (worst testing cost) of a decision tree with minimum possible expected testing cost (worst testing cost) over the instance 
$I.$ When the instance $I$ is clear from the context, we will also use the notation $OPT_W(S)$ ($OPT_E(S)$) for the above quantity, 
referring only to the set of objects involved.
We use $p_{min}$ to denote the smallest non-zero  probability among the objects in $S$.

Let $(S,T,C,{\bf p}, {\bf c})$ be an instance
of DFEP and let  $S'$ be a subset of $S$.
In addition, let $C', \, \mathbf{p}'$ and $\mathbf{c}'$ be, respectively, the restrictions of $C, \, \mathbf{p}$
and ${\bf c}$ to the set $S'$.
Our first observation is that every decision tree
$D$ for  $(S,C,T,\mathbf{p}, {\bf c})$ is also a decision tree
for the instance $I'=(S',C',T,\mathbf{p}', {\bf c}')$. 
The following proposition immediately follows.

\begin{proposition}
\label{prop:Subadditivity}
Let $I=(S,C,T,\mathbf{p}, {\bf c})$ be an instance
of the DFEP and let  $S'$ be a subset of $S$.
Then, $OPT_E(I') \leq OPT_E(I)$ and
$OPT_W(I') \leq OPT_W(I),$ where $I'=(S', C',T,\mathbf{p}', {\bf c}')$ is the restriction of 
$I$ to $S'$.
\end{proposition}


One of the measures of progress of our strategy is expressed in terms of the number of pairs of objects belonging to different classes which are present in the set of objects satisfying 
the tests already performed. The following definition formalizes this concept of pairs for a given set of objects.  
\begin{definition}[Pairs]
Let $I=(S,T,C,{\bf p}, {\bf c})$ be an instance of the DFEP and  $G \subseteq S.$  
We say that two objects $x, y \in S$ constitute a pair  of $G$ if they both belong to $G$ but come from different classes.
We denote by $P(G)$ the number of pairs of $G.$ In formulae, we have
$$P(G) = \sum_{i=1}^{m-1} \sum_{j=i+1}^m n_i(G) n_j(G)$$ 
where for $1\leq i \leq m$  and $A \subseteq S,$  $n_{i}(A)$ 
denotes the number of objects in $A$ belonging to class $C_{i}.$
\end{definition}
As an example, for the set of objects $S$ in Figure \ref{fig:decisiontree0} we have $P(S) = 8$ and the following set of pairs 
$\displaystyle{\{(1,3),(1,4),(1,5),(2,3),(2,4),(2,5),(3,4),(3,5)\}}.$ 

We will use $s^*$ to denote the initially unknown object whose class we want to identify.
Let ${\bf t}$ be a sequence of tests  applied to identify the class of $s^*$ (it corresponds to a path in the decision tree) and let $G$ be the set of objects that agree with the outcomes of all tests in ${\bf t}$.
If $P(G) = 0$,
then all objects in $G$  belong to the same class, which  must coincide with the class of the selected object $s^*$. 
Hence, $P(G) = 0$ indicates the identification of the class of the object $s^*.$ Notice that $s^*$ might still be unknown when the condition 
$P(G) = 0$ is reached.

For each  test $t \in T$ and for each $i=1,\ldots,\ell$, let $S^i_t \subseteq S$ be the set of objects for which the outcome of test $t$ is $i.$
For a test $t,$ the outcome resulting in the largest number of pairs is of special interest for our strategy. 
We denote with  $S^*_t$ the set among $S^1_t, \dots, S^{\ell}_t$ such that $P(S^*_t)=\max\{P(S^1_t),\ldots, P(S^{\ell}_t)\}$ (ties are broken arbitrarily).
We denote with $\sigma_S(t)$ the set of objects not included in $S^*_t,$ i.e., we define $\sigma_S(t)= S \setminus S^*_t$.
Whenever $S$ is clear from the context we use $\sigma(t)$ instead of $\sigma_S(t)$.

Given a set of objects $S$, each test produces a tripartition of the pairs in $S$: the ones with both objects in $\sigma(t),$ those with
both objects in $S^*_t$ and those
with one object in  $\sigma(t)$ and one object in $S^*_t.$
We say that the pairs in  $\sigma(t)$ are {\em kept} by $t$ and the pairs with one object from 
$\sigma(t)$ and one object from $S^*_t$ are {\em separated} by $t.$
We also say that  a pair is {\em covered} by the test $t$ if it is either kept or separated by $t.$
Analogously, we say that a test $t$ covers an object $s$ if $s \in \sigma(t)$.


For any set of objects $Q \subseteq S$  the probability of $Q$ is $p(Q) = \sum_{s \in Q} p(s).$

\section{Logarithmic approximation for the Expected Testing Cost and the Worst Case Testing Cost} \label{sec:expcost}

In this section, we describe our algorithm {\tt DecTree} and analyze its performance. 
The concept of the separation cost of a sequence of tests will turn useful for defining and analyzing our algorithm.

\noindent
{\bf The separation cost of a sequence of tests.} 
Given an instance $I = (S, C, T, {\bf p}, {\bf c})$ of the DFEP, for a sequence of tests ${\bf t} = t_1, t_2, \dots, t_q,$ we define the 
separation cost of ${\bf t}$ in the instance $I,$ denoted by 
$sepcost(I, {\bf t}),$ as follows:
%
Fix an object $x.$ If there exists $j < q$ such that $x \in \sigma(t_j)$ then we set $i(x) = \min\{j \mid x \in \sigma(t_j)\}.$  
If $x  \not \in \sigma(t_j)$ for each $j=1, \dots, q-1,$ then we set  $i(x) = q.$ 
Let  $sepcost(I, {\bf t}, x)  = \sum_{j=1}^{i(x)} c(t_j)$ denote  the {\em cost of separating $x$ in the instance $I$ by means of the sequence} ${\bf t}.$ 
Then, the {\em separation cost of ${\bf t}$} (in the instance $I$) is defined by 
\begin{equation} \label{eq:sepcost2}
sepcost(I,{\bf t}) = \sum_{s \in S} p(s) sepcost(I, {\bf t}, s).
\end{equation}

In addition, we define $totcost(I, {\bf t})$ as the total cost of the sequence ${\bf t}$, i.e., 
$$totcost(I, {\bf t}) = \sum_{j=1}^q c(t_j).$$

\noindent
{\bf Lower bounds on the cost of an optimal decision tree for the DFEP.}
We denote by  $sepcost^*(I)$  the minimum separation cost in $I$ attainable by a 
sequence of tests in $T$ which covers all the pairs in $S$
and $totcost^*(I)$ as the minimum total cost  attainable by a 
sequence of tests in $T$ which covers all the pairs in $S.$

The following theorem shows lower bounds on both the expected testing cost and the worst case testing cost of  any instance $I = (S, C, T, {\bf p}, {\bf c})$ of the DFEP.

\begin{theorem} \label{theo:lowerbound}
For any instance $I = (S, C, T, {\bf p}, {\bf c})$  of the DEFP, it holds that 
$sepcost^*(I) \leq OPT_E(I)$ and  $totcost^*(I) \leq OPT_W(I).$
\end{theorem}
\begin{proof}
Let $D$ be a decision tree for the instance $I$.
Let $ t_1, t_2, \dots, t_q, l$ be the nodes in the root-to-leaf path in $D$ such that for each $i = 2, \dots, q,$ the node $t_i$ is 
on the branch stemming from $t_{i-1}$ which is associated with $S^*_{t_{i-1}}$, and the leaf node $l$ is the child of
$t_q$ associated with the objects in  $S^*_{t_q}.$

Let ${\bf t} = t_1, t_2, \dots, t_q$.
Abusing notation let us now denote with $t_i$ the test 
associated with the node $t_i$ so that   ${\bf t} $ is a sequence
of tests.
In particular, ${\bf t}$ is the sequence of tests performed according to the strategy defined by $D$ when the 
object $s^*$ whose class we want to identify, is such that  $s^* \in S^*_t$ holds for 
each test $t$ performed in the sequence. 

Notice that, by construction,  ${\bf t}$ is a sequence of tests covering all pairs of $S$.

\medskip

\noindent
{\em Claim.} For each object $s$ it holds that $sepcost(I, {\bf t}, s) \leq cost(D, s).$ 

If for each $i=1, \dots, q,$ we have that $s \not \in \sigma(t_i)$ then it holds that $cost(D, s) = \sum_{j=1}^q c(t_j) = sepcost(I,{\bf t}, s).$
Conversely, let $t_i$ be the first test in ${\bf t}$ for which $s \in \sigma(t_i).$ 
Therefore, we have that $t_1, t_2, \dots, t_i$ is a prefix of the root to leaf path followed when $s$ is the object chosen.
It follows that $cost(D, s) \geq \sum_{j=1}^i c(t_j) = sepcost(I, {\bf t}, s).$
The claim is proved.

\medskip
In order to prove the first statement of the theorem, we let $D$ be a decision tree which achieves the minimum possible expected cost, i.e., $cost_E(D) = OPT_E(I).$
Then, we have 

\begin{equation}
sepcost^*(I) \leq sepcost(I, {\bf t}) = \sum_{s \in S} p(s) sepcost(I, {\bf t}, s)
\leq  \sum_{s \in S} p(s) cost(D, s) = OPT_E(I).
\end{equation}

In order to prove the second statement of the theorem, 
we let $D$ be a decision tree which achieves the minimum possible worst testing cost, i.e., $cost_W(D) = OPT_W(D).$
Let $s \in S$ be such that, for each $j = 1, \dots, q-1,$ it holds that $s \not \in \sigma(t_j).$ Then, by the above claim it follows that
\begin{equation} \label{eq:totcost-worstcase}
totcost(I, {\bf t}) = sepcost(I, {\bf t}, s) \leq cost(D, s) \leq cost_W(D).
\end{equation}
Using (\ref{eq:totcost-worstcase}), we have 

\begin{equation}
totcost^*(I) \leq totcost(I, {\bf t}) \leq cost_W(D) = OPT_W(I).
\end{equation} 

The proof is complete.
\end{proof}

The following subadditivity property will be useful.

\begin{proposition}[Subadditivity] \label{prop:subadditivity}
Let $S_1, S_2, \dots, S_q$ be a partition of the object set $S.$ We have 
$OPT_E(S) \geq \sum_{j=1}^q OPT_E(S_j)$ and $OPT_W(S) \geq \max_{j=1}^q \{ OPT_W(S_j) \}$, where 
$OPT_E(S_j)$ and $OPT_W(S_j)$ are, respectively, the minimum expected testing cost and the worst case testing cost when the set of objects is  $S_j.$
\end{proposition}


\noindent
{\bf The optimization of submodular functions of sets of tests.} \label{sec:SubmodularOtimization}
Let $I= (S, T, C, {\bf p}, {\bf c})$ be an instance of the DFEP. 
A set function $f:2^T \mapsto \mathbb{R}_+$ is submodular non-decreasing if for every 
$R \subseteq R' \subseteq T$ and every 
$t \in T \setminus R'$, it holds that  $f(R \cup \{t\}) -f(R) \geq f(R' \cup \{t\}) -f(R')$ (submodularity) and
$f(R) \leq f(R')$ (non-decreasing).

It is easy to verify that the functions $$f_1: R \subseteq T \mapsto P(S) - P(\bigcap_{t \in R} S^*_t)$$ 
 $$f_2 : R \subseteq T \mapsto p(S) - p(\bigcap_{t \in R} S^*_t)$$
are non-negative non-decreasing submodular set functions. 
In words, $f_1$ is the function mapping a set of tests $R$ into the number of pairs covered by the tests in $R$. The function $f_2$, instead, 
maps a set of tests $R$ into the probability of the set of objects covered by the tests in $R$.

Let $B$ be a positive integer. Consider the following  optimization problem defined over a  non-negative, non-decreasing, sub modular function $f$:
%
\begin{equation} 
{\cal P}(f,B,T,{\bf c}) :  \max_{R \subseteq T} \left \{ f(R) : \sum_{t \in R}   c(t) \leq B  \right \}.
\label{eq:problemP}
\end{equation}

In \cite{wolsey}, Wolsey  studied the solution to the problem ${\cal P}$ provided by Algorithm \ref{Greedy-wolsey} below, called  
the  adapted greedy heuristic. 

\begin{algorithm}[ht]
\small
{\bf Procedure} {\tt Adapted-Greedy}($S, T, f, {\bf c}, B$)~~~~
\begin{algorithmic}[1]

\STATE  {$spent \leftarrow 0, \; A \leftarrow \emptyset, \; k \leftarrow 0$} \label{line:one}
\REPEAT \label{line:two}
  \STATE{$k \leftarrow k+1$} \label{line:three}

  \STATE  {let $t_k$ be a test $t$ which maximizes  $\frac{f(A \cup \{t\})- f(A)}{c(t)}$ among all $t \in T$ s.t. $c(t) \leq B$} \label{line:four}

   \STATE {$T \leftarrow T \setminus \{t_k\}, \; spent \leftarrow spent + c(t_k), \; A \leftarrow A \cup \{t_k\}$} \label{line:five}

\UNTIL{$spent > B$ or $T=\emptyset$} \label{line:six}

\STATE{{\bf if} $f(\{t_k\}) \geq f(A \setminus \{t_k\})$ {\bf then} Return  $\{t_k\}$ \\ ~~~~~~~~{\bf else} Return $\{t_1\, t_2 \, \dots \, t_{k-1}\}$  \label{output}}
\end{algorithmic}
\caption{Wolsey greedy algorithm}
\label{Greedy-wolsey}
\end{algorithm}
\normalsize

The following theorem summarizes results from \cite{wolsey}[Theorems 2 and 3].
\begin{theorem} \cite{wolsey} \label{theo:wolsey}
Let $R^*$ be the solution of the problem ${\cal P}$ and
$\overline{R}$ be the set returned  by Algorithm \ref{Greedy-wolsey}. Moreover, 
let $e$ be the base of the natural logarithm and $\chi$ be the solution of $e^{\chi} = 2 - \chi.$ 
Then we have that $f(\overline{R}) \geq (1- e^{-\chi})f(R^*) \approx 0.35 f(R^*)$.
Moreover, if there exists $c$ such $c(t) = c$ for each $t \in T$ and $c$ divides $B,$ then 
we have 
$f(\overline{R}) \geq (1-1/e) f(R^*) \approx 0.63 f(R^*).$
\end{theorem}

\begin{corollary} \label{cor:wolsey}
Let ${\bf t} = t_1 \dots t_{k-1} t_{k}$ be the sequence of all the tests selected by {\tt Adapted-Greedy}, i.e., 
the concatenation of the two possible outputs in line 7. 
Then, we have that 
the total cost of the tests in ${\bf t}$ is at most $2B$ and $f(\{t_1, \dots, t_{k-1}, t_{k}\}) \geq 
(1-e^{-\chi}) f(R^*) \approx 0.35 f(R^*).$ 
\end{corollary}

Our algorithm for building a decision tree  will employ this greedy heuristic for finding approximate solutions to  the optimization problem ${\cal P}$ over the submodular 
set  functions $f_1$ and $f_2$ defined in (\ref{eq:problemP}).

\subsection{Achieving logarithmic approximation}
We will show that Algorithm
\ref{algo:main}
 attains a logarithmic approximation for 
DFEP.
The algorithm consists of  4 blocks.
The first block (lines 1-2)
is the basis of the recursion,
which  returns a leaf if all objects belong to the same class $(P(S)=0)$.
If $P(S)=1$, we have that $|S|=2$ and the algorithm returns a tree that consists
of a  root and two leaves,  one for  each object, where the root is associated with  the cheapest test that separates these two objects. Clearly, this tree is optimal for
both the expected testing cost and the worst testing cost.

The second block (line 3)
calls procedure {\tt FindBudget} to define
the budget $B$ allowed for the tests selected in the third and fourth blocks. 
{\tt FindBudget} finds the smallest $B$ such that 
{\tt Adapted-Greedy}($S,T, f_1, {\bf c}, B$)
returns  a set of tests $R$ covering  at least $\alpha P(S)$ pairs.

%
\begin{algorithm}[ht]
\small
{\bf Procedure} {\tt DecTree}($S,T,C,{\bf p}, {\bf c}$)
\begin{algorithmic}[1]
\STATE { {\bf If} $P(S)=0$ {\bf then return} a single leaf $l$ associated with $S$} \label{b1start}
\STATE { {\bf If} $P(S)=1$ {\bf then return} a tree whose root is the  cheapest test that separates the two objects in $S$} \label{b1end}
\STATE{  $B$={\tt FindBudget$(S,T, C, {\bf c})$}, \; $spent \leftarrow 0, \; spent_2 \leftarrow 0, \; U \leftarrow S, \; k\leftarrow 1$  } \label{line:budget} \label{b2start}
\WHILE{ there is a test in $T$ of cost  $\leq  B - spent$} \label{b3-start}
  \STATE  {let $t_k$ be a test which maximizes $\frac{p(U) - p(U \cap S^*_t )}{c(t)}$ among all tests $t $  s.t. \ $t \in T$ and  $c(t) \leq  B - spent$} \label{line:main-opt1-1}
  \STATE {{\bf If} $k=1$ {\bf then} make $t_1$ root of $D$ {\bf else} $t_k$ child of $t_{k-1}$}
  \FOR { every $i \in \{1,\ldots,\ell\}$ such that $(S^i_{t_k} \cap U) \neq \emptyset$ {\bf and} $ S^i_{t_k} \neq  S^*_{t_k}$} \label{line:main-Ui1}
 \STATE{ Make $D^i \leftarrow $ {\tt DecTree}($S^i_{t_k} \cap U,T,C,{\bf p}, {\bf c}$) 
   child of $t_k$} \label{line:main-Call-1}
 \ENDFOR
\STATE {$U\leftarrow U \cap S^*_{t_k}, \;  spent \leftarrow spent + c(t_k)$ \;, $T \leftarrow  T \setminus \{t_k\}$, \; $k\leftarrow k+1$} 
 \ENDWHILE
\STATE{\bf end while} \label{b3-end}
\REPEAT \label{b4-start}
  \STATE  {let $t_k$ be a test which maximizes $\frac{P(U)- P(U \cap S^*_t)}{c(t)}$   among all tests $t \in T$ s.t. $c(t)\leq B$} \label{line:main-greedy2}
  \STATE {Set $t_k$ as a child of $t_{k-1}$}
  \FOR { every $i \in \{1,\ldots,\ell\}$ such that $(S^i_{t_k} \cap U) \neq \emptyset$ {\bf and} $ S^i_{t_k} \neq  S^*_{t_k}$} \label{line:main-Ui2}
   \STATE {Make $D^i \leftarrow ${\tt DecTree}($U \cap S^i_{t_k},T,C,{\bf p}, {\bf c}$) child of $t_k$}\label{line:main-Call-2}
 \ENDFOR
\STATE {$U\leftarrow U \cap S^*_{t_k}, \;  spent_2 \leftarrow spent_2 + c(t_k)$ \;, $T \leftarrow  T \setminus \{t_k\}$, \; $k\leftarrow k+1$}	
\UNTIL{$B - spent_2 < 0$ {\bf or} $T = \emptyset$} \label{b4-end}
\STATE{$D' \leftarrow $ {\tt DecTree}($U,T, C, {\bf p}, {\bf c}$); Make $D'$ a child of $t_{k-1}$  } \label{line:main-Call-3}
\STATE{Return the decision tree $D$ constructed by the algorithm}
\end{algorithmic}
\bigskip
{\bf Procedure} {\tt FindBudget}($S,T,C,{\bf c}$)
\begin{algorithmic}[1]
\STATE { Let $f : R \subseteq T \mapsto P(S) - P(\bigcap_{t \in R} S^*_t)$  and let $\alpha=1-e^{\chi}\approx 0.35$}
\STATE { Do a binary search in the interval [1,$\sum_{t \in T} c(t)$] to find the smallest $B$ such that 
{\tt Adapted-Greedy}($S,T, f, {\bf c}, B$)
returns  a set of tests $R$ covering  at least $\alpha P(S)$ pairs}
\STATE { Return $B$}
\end{algorithmic}
\caption{}  \label{algo:main}
\end{algorithm}
\normalsize

The third (lines 4-10)
and the fourth (lines 11-17)
blocks
are responsible for the construction of  the backbone of the decision tree (see Fig.\ 2) 
as well as to call {\tt DecTree} recursively to construct  the decision trees that are children of the nodes in the backbone.

The third block (the {\bf while} loop in lines 4-10) constructs the first part of the backbone 
(sequence ${\bf t}^A$ in Fig.\ 2) 
by iteratively selecting
the test that covers the maximum uncovered mass probability per unit of testing cost (line 5).
The selected test $t_k$ induces a partition $(U_{t_k}^1,\ldots,U_{t_k}^{\ell})$ on the set of objects $U$,
which contains the objects that have not been covered yet.
In lines 7 and 8, the procedure is recursively called for each set of this partition but for the one that is contained in the 
subset $S^*_{t_k}$. 
With reference to Figure 2, these calls will build the subtrees rooted at nodes not in ${\bf t}^A$ which are children of some node in ${\bf t}^A$.


Similarly, the fourth block (the {\bf repeat-until} loop) constructs the second part of the backbone 
(sequence ${\bf t}^B$ in Fig.\ 2) 
 by  iteratively selecting
the test that covers the maximum number of uncovered pairs per unit of testing cost (line 12). 
The line \ref{line:main-Call-3} is
responsible
for building a decision 
tree for the objects that are not covered
by the tests in the backbone.

We shall note that both the  third and the fourth block of the algorithm  are 
based on the adapted greedy heuristic of Algorithm \ref{Greedy-wolsey}. 
In fact,  $p(U)-p(U\cap S_{t_k}^*) $ in line \ref{line:main-opt1-1} (third block) 
 corresponds to $f_2(A \cup t_k)-f_2(A) $ in Algorithm \ref{Greedy-wolsey} because,
right before the selection of the $k$-th test,
$A$ is the set of  tests  $\{t_1 ,\ldots, t_{k-1}\}$ and $U=\cap_{i=1}^{k-1} S^*_{t_i}$.
Thus,  $$f_2(A \cup t_k )=p(S)-p(\cap_{i=1}^k S^*_{t_i})=p(S)- p(U \cap S^*_{t_k}) $$ 
 and  $$f_2(A )=p(S)-p(\cap_{i=1}^{k-1} S^*_{t_i})=p(S)- p(U)$$
so that
 $$f_2(A \cup t_k ) - f_2(A ) = p(U)-p(U\cap S_{t_k}^*).$$

A similar argument shows that $P(U)-P(U\cap S_t^*) $ in line \ref{line:main-greedy2} (fourth block)
corresponds to $f_1(A \cup t_k)-f_1(A) $ in Algorithm \ref{Greedy-wolsey}.
These connections will allow us to  apply both Theorem \ref{theo:wolsey} and Corollary \ref{cor:wolsey} to analyze the cost and the coverage of these sequences.

\begin{figure}[h]
\centering
\includegraphics[width=80mm]{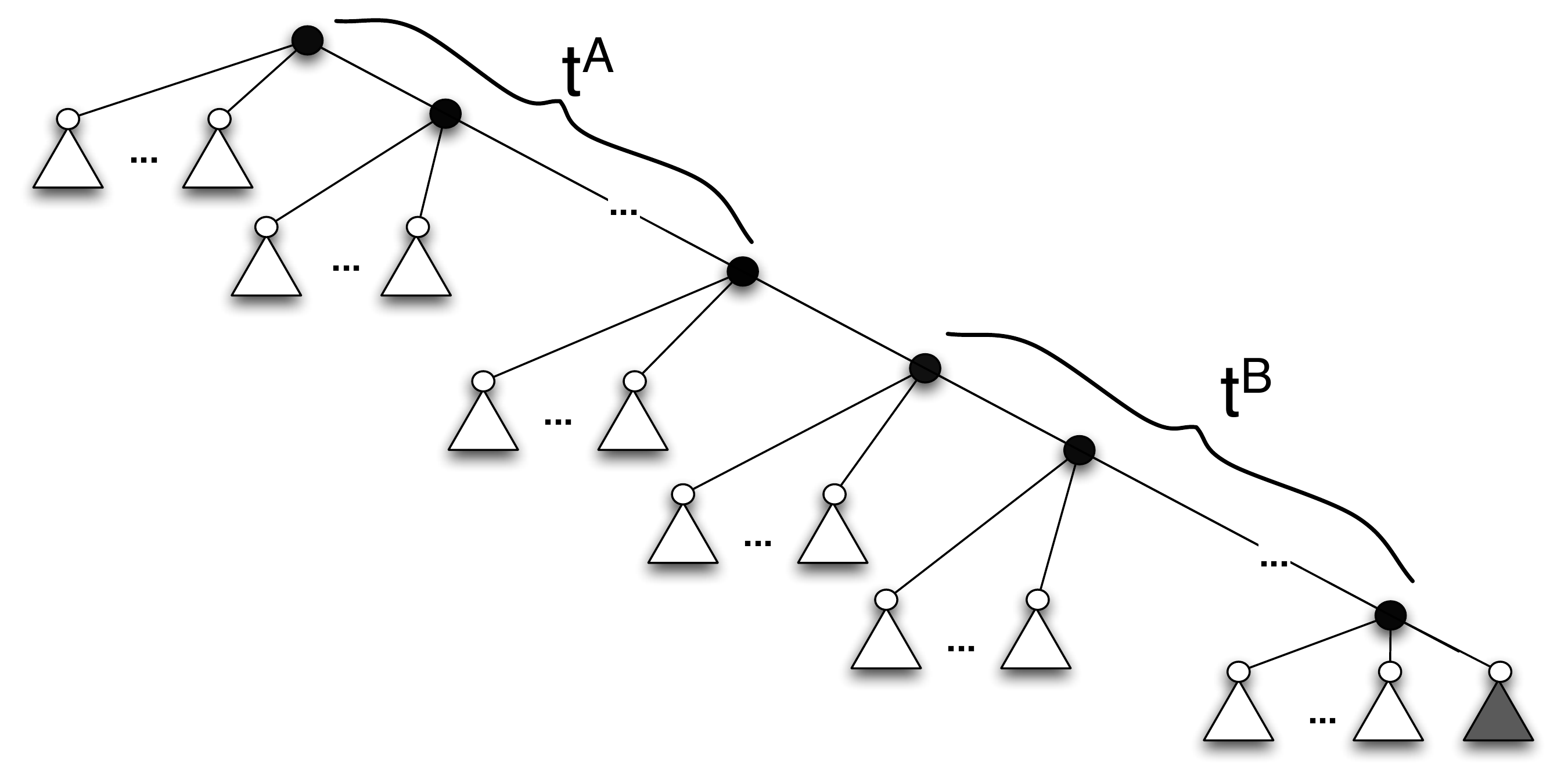}
\label{fig:decisiontree}
  \caption{The structure of the decision tree built by {\tt DecTree}: white nodes correspond to recursive calls. In each white subtree, the number of pairs is at most $ P(S)/2,$ while in the lowest-right gray subtree it is at most $8/9 P(S)$ (see the proof of Theorem \ref{theo:main}).}
\vskip -0.1in
\end{figure}

Let ${\bf t}_I$ denote the sequence of tests
obtained by concatenating the tests
selected in the {\bf while} loop  and in the {\bf repeat-until} loop  of the execution of {\tt DecTree} over instance $I.$
We delay to the next section the proof of the following key result.

\begin{theorem} \label{theo:key}
Let $\chi$ be the solution of $e^{\chi} = 2-\chi,$ and $\alpha =1-e^{\chi} \approx 0.35.$ 
There exists a constant $\delta \geq 1,$ such that for any instance $I = (S, C, T, {\bf p}, {\bf c})$ of the DFEP,
the sequence ${\bf t}_I$ covers at least $\alpha^2 P(S) > \frac{1}{9} P(S)$ pairs,  and it holds that 
$sepcost(I, {\bf t}_I) \leq \delta \cdot sepcost^*(I)$ and $totcost(I, {\bf t}_I) \leq 3 totcost^*(I).$
\end{theorem}

Applying Theorem \ref{theo:key} to each recursive call of {\tt DecTree} we can prove the following theorem about the approximation guaranteed 
by our algorithm both in terms of worst testing cost and expected testing cost.

\begin{theorem} \label{theo:main}
For any instance $I=(S, C, T, {\bf p}, {\bf c})$ of the DFEP,  the algorithm {\tt DecTree} outputs a decision tree with expected testing cost at most $O(log(n)) \cdot OPT_E(I)$ and with worst testing cost   at most $O(log(n)) \cdot OPT_W(I)$.
\end{theorem}

\begin{proof}
For any instance $I,$ let $D^{\mathbb{A}}(I)$ be the decision tree produced by the algorithm {\tt DecTree}.  
First, we prove an approximation for the expected testing cost.
Let $\beta$ be such that $ \beta \log \frac{9}{8}= \delta $, where $\delta$ is the constant given in the statement of Theorem \ref{theo:key}.
Let us assume by induction that the algorithm guarantees approximation 
$1+\beta \log P(G) $, for the expected testing cost, for  every instance $I'$ on a set of objects $G$ with $ 1 \leq P(G) < P(S).$ 

Let ${\cal I}$ be the set of instances
on which the algorithm ${\tt DecTree}$ is recursively called in lines
8,15 and 18.
We have that 
%
%
\begin{eqnarray}
\!\!\!\!\!\! \frac{cost_E(D^{\mathbb{A}}(I))}{OPT_E(I)} \!\!\!\!\! &=& \!\!\!\! \frac{sepcost(I, {\bf t}_I) + \sum_{I' \in {\cal I}} cost_E(D^{\mathbb{A}}(I'))}{OPT_E(I)} \label{1}\\
&\leq& \!\!\!\! \frac{sepcost(I, {\bf t}_I)}{OPT_E(I)} + \max_{I' \in {\cal I}}\frac{cost_E(D^{\mathbb{A}}(I'))}{OPT_E(I')} \label{3} \\
&\leq& \!\!\!\! \delta + \max_{I' \in {\cal I}}\frac{cost_E(D^{\mathbb{A}}(I'))}{OPT_E(I')} \label{3-1}\\
&\leq& \!\!\!\! \delta + \max_{I' \in {\cal I}} \left \{1+\beta \log P(I') \right \} \label{4}\\
&\leq& \!\!\!\! \delta +  1+ \beta \log 8P(S)/9 = 1+ \beta log(P(S)). \label{6} 
\end{eqnarray}

The first equality follows by the recursive way the algorithm {\tt DecTree} builds the decision tree.
Inequality (\ref{3}) follows from (\ref{1}) by the subadditivity property (Proposition  \ref{prop:subadditivity}) and
simple algebraic manipulations.  
The inequality in (\ref{3-1}) follows  by Theorem \ref{theo:key} together with Theorem \ref{theo:lowerbound} yielding 
$sepcost(I, {\bf t}_I) \leq \delta \, OPT_E(I).$
The  inequality (\ref{4})  follows by induction (we are 
using $P(I')$ to denote the number of pairs of instance $I'$).

To prove that the  inequality in (\ref{6}) holds
we have to argue that every instance $I'\in {\cal I}$ has at most $\frac{8}{9} P(S)$ pairs.
Let $U_{t_k}^i = S_{t_k}^i \cap U$ as in the lines 8 and 15. First we show that  the number of pairs  of  $U_{t_k}^i$ 
is at most $ P(S)/2$.
We have $S_{t_k}^i \neq S_{t_k}^*$  and  $S^*_{t_k}$ is the 
set with the maximum number of pairs in the partition $\{S^1_{t_k},\ldots,  S^{\ell}_{t_k}\}$,
induced by $t_k$ on the set $S$.  It follows
that $ P(U_{t_k}^i) \leq P(S_{t_k}^i) \leq P(S)/2.$
Now it remains to show that the instance $I'$, recursively called, in line 18 has at most
 $8/9 P(S)$ pairs. This is true because the number of pairs of $I'$ is equal to the number of pairs not covered by ${\bf t}_I$ which is bounded by $(1-\alpha^2) P(S) \leq 8P(S)/9$ by Theorem \ref{theo:key}.

Now, we prove an approximation for the  worst testing cost of the tree $D^{\mathbb{A}}(I)$.
Let $ \rho $ be such that $ \rho \log \frac{9}{8} = 3$. 
Let  us assume by induction  that the worst testing cost of 
$D^{\mathbb{A}}(I')$
is at most  $(1+ \rho \log P(G)  \cdot OPT_W(I')) $ for  every instance $I'$ on a set of objects $G$ with $1 \leq P(G) < P(S).$ 
We have that 
%
%
\begin{eqnarray}
\!\!\!\!\!\! \frac{cost_W(D^{\mathbb{A}}(I))}{OPT_W(I)} \!\!\!\!\! & \leq & 
\frac{ totcost(I,{\bf t}_I) + \max_{I' \in {\cal I}} \{ cost_W(D^{\mathbb{A}}(I')) \} } { OPT_W(I) } \\
&\leq& \!\!\!\! \frac{totcost(I,{\bf t}_I)}{OPT_W(I)} + \max_{I' \in {\cal I}}  
\frac{ cost_W(D^{\mathbb{A}}(I'))}{ OPT_W(I') } \label{8} \\
&\leq& \!\!\!\! \frac{totcost(I,{\bf t}_I)}{totcost^*(I)} + \max_{I' \in {\cal I}}  
\frac{ cost_W(D^{\mathbb{A}}(I'))}{ OPT_W(I') } \label{9} \\
&\leq& \!\!\!\! 3 + 1 + \rho \log (8 P(S)/9) = 1+ \rho \log (P(S)) \label{10} \\ \nonumber \end{eqnarray}

Inequality (\ref{8}) follows from   the subadditivity property (Proposition  \ref{prop:subadditivity}) for the worst testing cost.
The inequality  (\ref{9}) follows  by  Theorem \ref{theo:lowerbound}.
The   inequality  (\ref{10})  follows from Theorem \ref{theo:key}, the  induction hypothesis (we are 
using $P(I')$ to denote the number of pairs of instance $I'$) and from the fact mentioned above that every instance
in ${\cal I}$ has at most $8/9P(S)$ pairs.

Since $P(S) \leq n^2$ it follows that the algorithm provides
an $O(\log n)$ approximation for both the expected testing cost and the worst testing cost.

\end{proof}

The previous theorem shows  that algorithm {\tt DecTree}
provides  simultaneously logarithmic  approximation
for the minimization of expected testing cost and worst testing cost.
We would like to remark  that this is an  interesting feature of our algorithm. In this respect,  
let us consider the following instance of the DFEP\footnote{ 
This is also an instance of the identification problem mentioned in the introduction}:
Let $S=\{s_1,\ldots,s_n\}$; $p_i=2^{-i}$, for $i=1,\ldots,n-1$ and $p_n=2^{-(n-1)}$; 
the set of tests is in one to one correspondence with   the set of all binary strings
of length $n$ so that  the test corresponding to a binary string ${\bf b}$ outputs
$0(1)$ for object $s_i$ if and only if the $i$th bit of ${\bf b}$ is 0(1).
Moreover, all tests have unitary costs.
This instance is  also an instance of the problem of constructing
an optimal prefix coding binary  tree, which can be solved by the Huffman's algorithm \cite{cormen-algorithms}.
Let $D^*_E$ and  $D^*_W$ be, respectively, the decision trees with minimum expected cost and
minimum worst testing cost for this example.
Using Huffman's algorithm, it is not difficult to verify  that 
 $Cost_E(D^*_E) \leq 3$ and  $Cost_W(D^*_E)=n-1$. In addition, 
we have that $Cost_E(D^*_W)=Cost_W(D^*_W)=\log n$.
This  example shows that the minimization of the expected testing
cost may result in high worst testing cost and vice versa the 
minimization of the worst testing
cost may result in high expected  testing cost. 
Clearly, in real situations presenting such a dichotomy, 
the ability of our algorithm to optimize simultaneously both measures of cost
might provide a  significant gain over strategies only guaranteeing competitiveness with respect to one measure.

\subsection{The proof of Theorem \ref{theo:key}}

We now return to the proof of Theorem \ref{theo:key} for which will go through three lemmas. 

\begin{lemma} \label{lemma:length}
For any instance $I = (S, C, T, {\bf p}, {\bf c})$ of the DFEP,
the value $B$ returned by the procedure  {\tt FindBudget}$(S,T,C,{\bf c})$ satisfies 
 $B \leq totcost^*(I)$.
\remove{In addition, the execution  of procedure {\tt Adapted-Greedy} over
the instance $I$, with budget $B$, and $f : R \subseteq S \mapsto P(S) - P(\bigcap_{t \in R} S^*_t)$ selects a sequence of tests that  covers at least
$\alpha P(S)$ pairs, where $\alpha= 1-e^{\chi} \approx 0.35$ and $\chi$ is the 
solution of $e^{\chi} = 2-\chi.$}
\end{lemma}
\begin{proof}
Let us consider the problem ${\cal P}$ in equation (\ref{eq:problemP})  with the function $f_1$ that
measures the number of pairs covered by a set of tests.
Let $G(x)$ be the number of pairs covered by the solution constructed with {\tt Adapted-Greedy}
when the budget---the righthand side of equation (\ref{eq:problemP})---is $x$. 
By construction, {\tt FindBudget} finds the smallest $B$ such that $G(B) \geq \alpha P(S)$.

Let $\tilde{\bf t}$ be a  sequence that covers  all  pairs in $S$ and that satisfies 
$totcost( \tilde{\bf t}) = totcost^*(I)$.
 Arguing by contradiction we can show that $totcost(I, \tilde{\bf t}) \geq B.$ Suppose that this was not the case, then
$\tilde{\bf t}$ would be the sequence which covers  $P(S)$ pairs using a sequence of tests of total cost not larger than some $B' < B.$
By Theorem 2,  the procedure {\tt Adapted-Greedy} provides an $\alpha$-approximation of the maximum number of pairs covered 
with a given budget.  Therefore, when run with budget $B',$ {\tt Adapted-Greedy} is guaranteed to produce a sequence of total cost $\leq B'$ which covers 
at least $\alpha P(S)$ pairs. 
However,  by the minimality of $B$ it follows that such a sequence does not exist. 
Since this contradiction follows by the hypothesis  $totcost(I, \tilde{\bf t}) < B,$ it must hold that $totcost^*(I)  \geq totcost(I, \tilde{\bf t}) \geq B,$ as desired.
\end{proof}

Given an instance $I,$ for a sequence of tests ${\bf t}=t_1,\ldots,t_k$ and a real $K>0$,
let $sepcost_K(I, {\bf t})$ be
the separation cost of ${\bf t}$ 
when every non-covered object is charged  $K$, that is,
$$sepcost_K(I, {\bf t}) =\sum_{\substack{x \in S \\ x \mbox{ is covered by } {\bf t}} } p(x)sepcost(I,{\bf t}, x)+
 \sum_{\substack{x \in S \\ x \mbox{ is not covered by } {\bf t}}} p(x) \cdot K $$

The proofs of the following technical lemma is deferred to the appendix.

\begin{lemma} \label{lemma:key2} 
Let  
${\bf t}^A$ be the sequence obtained by concatenating the tests selected in the {\bf while} loop of Algorithm 
\ref{algo:main}.
Then,  $totcost(I, {\bf t}^A) \leq B$ and
 $sepcost_B(I, {\bf t}^A) \leq \gamma \cdot  sepcost^*(I),$
where  $\gamma$ is a positive constant and $B$ is the budget calculated at line 3.
\end{lemma}

\begin{lemma} \label{lemma:key3} 
The sequence ${\bf t}_I$ covers at least $\alpha^2 P(S)$ pairs and it holds that 
$totcost(I, {\bf t}_I) \leq 3B.$
\end{lemma}
\begin{proof}
The sequence ${\bf t}_I$ can be decomposed
into the sequences ${\bf t}^A$ and ${\bf t}^B$,
that are constructed, respectively, in the {\bf while} and {\bf repeat-until}
loop of the algorithm {\tt DecTree}  (see also Fig. 2).

It follows  from the definition of $B$ that there is a sequence of tests, say ${\bf t}$,  of total cost not larger than $B$ 
that covers at least $\alpha P(S)$ pairs for instance $I$. 
Let $z$ be the number of pairs of instance $I$ covered by the sequence ${\bf t}^A$.
Thus, the tests in ${\bf t}$, that do not belong to ${\bf t}^A$, cover
 at least $\alpha P(S) - z $ pairs in the set $U =\bigcap_{t \in {\bf t}^A} S^*_t$ of  objects not covered by ${\bf t}^A$.
 
 The sequence ${\bf t}^B$ coincides with the concatenation of the two possible outputs of the procedure 
 {\tt Adapted-Greedy}($U, T-{\bf t}^{A}, f', {\bf c}, B)$ (Algorithm 1), when it is executed on 
 the instance defined by:  the objects in $U$ (those not covered by ${\bf t}^A$);  the 
 tests that are not in ${\bf t}^A;$ the submodular  set function $f': R \subseteq T - {\bf t}^A \mapsto P(S) - P( U \cap \left ( \bigcap_{t \in R} S_t^* \right ) )$ and 
 bound $B.$ By Corollary \ref{cor:wolsey}, we have that $totcost(I, {\bf t}^B) \leq 2B$ and 
 ${\bf t}^B$ covers at least  $ \alpha( \alpha P(S) - z )$ uncovered pairs.
 
Therefore, since $totcost(I, {\bf t}^A) \leq B,$ altogether, we have that 
${\bf t}_I$ covers at least $ z+ \alpha( \alpha P(S) - z ) \geq \alpha^2 P(S)$ pairs and $totcost(I, {\bf t}_I) \leq 3B.$
\end{proof}

The proof of Theorem \ref{theo:key} will now follow by combining the previous three lemmas.

\medskip

\noindent
{\bf Proof of Theorem \ref{theo:key}.}
First, it follows from Lemma \ref{lemma:key3} that
${\bf t}_I$ covers at least $\alpha^2 P(S)$ pairs.

To prove that  $sepcost(I, {\bf t}_I)\leq sepcost^*(I)$, we decompose ${\bf t}_I$ into 
 ${\bf t}^A= t_1^A,\ldots,t^A_q$ and ${\bf t}^B=t_1^B,\ldots,t^B_r$,
 the sequences of tests selected in the {\bf while } and in the {\bf repeat-until }
 loop of Algorithm \ref{algo:main}, respectively.

For $i=1,\ldots,q$, let $\pi_i = \sigma(t^A_i) \setminus ( \bigcup_{j=1}^{i-1} \sigma(t^A_j))$.
In addition, let $\pi_A$ be the set of objects which are not covered  by the tests in ${\bf t}^A.$
Thus,
\begin{eqnarray*}
sepcost(I, {\bf t}_I) &\leq& \sum_{i=1}^q p(\pi_i) \left ( \sum_{j=1}^{i} c(t^A_j) \right ) + 3B \cdot p(\pi_A)\\
&\leq&  3sepcost_B(I, {\bf t^A}) \leq 3 \gamma sepcost^*(I),
\end{eqnarray*} 
where the last inequality follows from Lemma \ref{lemma:key2}.

It remains to show that $totcost(I, {\bf t}_I) \leq 3 totcost^*(I)$.
This inequality holds because Lemma \ref{lemma:key3} assures that 
$ totcost(I, {\bf t}_I) \leq 3B$ and Lemma \ref{lemma:length} assures that 
$totcost^*(I) \geq B$.
The proof is complete.

\section{$O(\log n)$ is the best possible approximation.}  \label{sec:inapprox}
Let $U=\{u_1,\ldots,u_n\}$ be a set of $n$ elements and
${\cal F}$ be a family of subsets of $U$.
The minimum set cover problem asks for a family ${\cal F}' \subseteq {\cal F}$ of minimum cardinality
such that $\bigcup_{F \in {\cal F}'} F = U$.
It is known that no sub logarithmic approximation is achievable for the minimum set cover problem under the standard 
assumption that  $P \neq NP.$ More precisely, by the result of Raz and Safra \cite{Raz-Safra} it follows that 
there exists a constant $\tilde{k} > 0$ such that no 
$\tilde{k} \, log_2 n$-approximation algorithm for the  minimum set cover problem  exists
unless $P = NP$ \cite{Raz-Safra,Feige-98}.

We will show that  an $o( \log n)$ approximation algorithm for the 
 minimization of the expected testing cost  DFEP with exact $b$ classes for $b \geq 2$,  implies that
 the same approximation can be achieved for the 
Minimum Set Cover problem. This implies  that  one cannot expect to obtain a sublogarithmic approximation
for the DFEP unless $P=NP$. The reduction we present can also be used to show the same inapproximability result for the minimization of the 
worst testing cost version of the DFEP.

Given an instance $I_{SC} = (U, {\cal F})$ for the minimum set cover problem as defined above, we construct an instance 
$I_{DFEP}=(S, C, T, {\bf p}, {\bf c})$ for the DFEP as follows:
The set of objects is $S = U \cup \{o_1,\ldots,o_{b-1}\}$. The family of classes $C = (C_0, \dots, C_{b-1})$ is defined as follows:
All the objects of  $U$ belong to class $C_0$ while the object $o_i$, for $i=1,\ldots,b-1$,  belongs
to class $C_i$. Notice that $b_1$ and the objects of
$U$ belong to the same class. In order to define the set of tests $T,$ we proceed as follows: For each set $F \in {\cal F}$ we create a test $t_F$ such that
$t_F$ has value $0$ for the objects in $F$ and value 1 for the remaining objects.
In addition, we create a test $\tilde{t}$ which has value $0$ for objects in $U$ and
value $i-1$ for object $o_i$ ($1 \leq i \leq b-1$). 
For our later purposes, we notice here that the test $\tilde{t}$ cannot distinguish between $b_1$ and the elements in $U$.

Each test has cost 1, i.e., the cost assignment {\bf c} is given by $c(t) = 1$ for each $t \in T.$   
Finally, we set the probability of $o_1$ to be equal to $1- (n+b-2)\eta$ and the probability of the other objects equal to $\eta$,
for some fixed $\eta < \frac{1}{2(n+b-2)}.$

Let $D^*$ be the  decision tree  with minimum expected testing cost  for $I_{DFEP}$ and 
let ${\cal F}^*=\{F_{1},\ldots,F_{h}\}$ be a minimum set cover for instance $I_{SC}= (U,{\cal F})$, where $h=|{\cal F}^*|$.

We first argue that $cost_E(D^*) \leq h+1$. In fact, we can construct a decision tree $D$ 
by putting the test $t_{F_1}$ associated with $F_{1}$ in the root of the tree, then
the test $t_{F_2}$ associated with $F_{2}$ as the child of $t_{F_1}$ and so on. 
Notice that, for $i=1, \dots, h-1$ we have that $t_{F_i}$ has two children, one is $t_{F_{i+1}}$ and the other is a leaf  mapping to the class $C_0.$
As for $t_F{h},$ one of its children in again a leaf mapping to $C_0$, the other child is set to the test $\tilde{t}$, whose children are all leaves.

The expected testing cost  of $D$ can be upper bounded by 
\begin{equation} \label{inapprox:1}
OPT_E(I_{DFEP}) = cost_E(D^*) \leq cost_E(D) \leq  (h+1) = OPT(I_{SC}) +1
\end{equation}
since we have $cost(D, s) = (h+1)$ for any $s \in \{o_1, \dots, o_{b-1}\}$ and  
$cost(D, s) \leq h+1$ for any $s \in  U$.

On the other hand, let $D$ be a decision tree for $I_{DFEP}$ and let $P$ be the path from the root of $D$ to the leaf 
where the object $o_1$ lies. It is easy to realize that the subsets associated with the tests on this path cover all the
elements in $U$---in fact these tests separate $o_1$ from all the other objects from $U.$
Let ${\cal T}$ be the solution to the set cover problem provided by the sets associated with the tests on the path $P.$
We have  that 
\begin{equation} \label{inapprox:2}
|{\cal T}| \leq cost(D, o_1)  \leq \frac{cost_E(D)}{1- \eta(n+b-2)} \leq   2 cost_E(D),
\end{equation}
in the last inequality we are using the fact that $\eta \leq \frac{1}{2(n+b-2)}.$ 
\bigskip

Now assume that there is an algorithm that  for any instance $I = (S, C, T, {\bf p}, {\bf c})$ of the DFEP can guarantee a solution 
 with approximation $\alpha \log |S|$ for some $\alpha < \tilde{k}/8.$ 
Therefore, given an instance $I_{SC} = (U, {\cal F})$ for set cover we can use this algorithm on the transformed instance $I_{DFEP}$ defined 
above, where $|S| = |U|+b-1.$ 
We obtain a  decision tree $D$ for $I_{DFEP}$ such that 
$$cost_E(D) \leq \alpha \log(n+b-1) OPT_E(I_{DFEP}) \leq \alpha (OPT(I_{SC})+1) \log (n+b) \leq 4 \alpha \log n OPT(I_{SC})$$
where we upper bound $OPT(I_{SC}) + 1 \leq 2 OPT(I_{SC})$ and $\log (n+b) \leq 2\log n.$

From $D$, as seen above we can construct a solution ${\cal T}$ for the set cover problem such that $|{\cal T}| \leq 2 cost_E(D).$
Hence, it would follow that ${\cal T}$ is an approximate solution for the set cover instance satisfying:
$$|{\cal T}| \leq 8 \alpha \log n OPT(I_{SC}) <  (\tilde{k} \, \log n) OPT(I_{SC})$$ which by the result of \cite{Raz-Safra} is not possible unless $P = NP.$

\bigskip

The same construction can be used for analyzing the case of the worst testing cost, in which case we have that (\ref{inapprox:1})
becomes $OPT_W(I_{DFEP}) \leq OPT(I_{SC})+1$ and (\ref{inapprox:2}) becomes $|{\cal T}| \leq cost_W(D),$ leading to the inapproximability
of the DFEP w.r.t. the worst testing cost within a factor of $\alpha \log n$ for any $\alpha <  \tilde{k}/4.$ Notice that an analogous result regarding the 
worst testing cost had been previously shown by Moshkov \cite{Moshkov1} based on the result of Feige \cite{Feige-98}.

Thus, we have the following theorem

\begin{theorem} \label{theo:inappr}
Both the minimization of the  worst case and  the expected case
of the DFEP don't admit an $o(\log n)$ approximation unless $P=NP.$ 
\end{theorem}

 \section{Final remarks and Future Directions}
\label{sec:experiment}
We presented a new algorithm for the discrete function evaluation problem, a generalization of 
 the classical Bayesian active learning also studied under the names of 
 Equivalence Class Determination Problem \cite{golovin} and Group Based Active Query Selection problem of \cite{bellala}.
Our algorithm  builds a decision tree which asymptotically matches simultaneously for the expected and the worst testing cost the best possible approximation 
achievable under standard complexity assumptions $P \neq NP.$
This way,  we close the gap left open by the previous $O(\log 1/p_{\min})$ approximation for the expected cost shown in\cite{golovin} and  \cite{bellala}, where  $p_{min}$ is the minimum positive probability among the objects in $S$ and in addition we show that this can be done with an algorithm that guarantees the best possible approximation also with respect to the worst testing cost.

With regards to the broader context of learning,  given a set of samples labeled according to an unknown function,  
a standard task in machine learning is to find a good approximation of the labeling function (hypothesis). 
 In order to guarantee that the hypothesis chosen has some generalization power w.r.t.\ to  the  set of samples, 
 we should avoid overfitting. When the learning is performed via decision tree induction this implies that 
 we  shall not have leaves associated with a small number of samples  so that we end up with a decision 
 tree that have leaves associated with more than one  label. 
 There are many strategies available in the literature  to induce such a  tree. 

In the problem considered in this paper our aim is to over-fit the data because the function is 
 known a priori  and we are interested in obtaining a decision tree that allows us to identify the label of 
a new sample with the minimum possible cost (time/money). 
The theoretical results we obtain for the "fitting" problem should be generalizable to the problem of approximating the function. 
To this aim we could employ the framework of covering and learning from \cite{conf/icml/GuilloryB11} along the following lines: 
 we would  interrupt the recursive process in Algorithm 2 through which we  construct the  
 tree as soon as we reach a certain level of learning (fitting) w.r.t.  the set of labeled samples. 
Then, it remains to  show that our decision tree is at logarithmic factor of the optimal one for that level of learning.
 This is an interesting direction for future research.

\newpage

\appendix

\section{The proof of Lemma \ref{lemma:key2}}

\noindent
{\bf Lemma \ref{lemma:key2}.} 
{\em Let  
${\bf t}^A$ be the sequence obtained by concatenating the tests selected in the {\bf while} loop of Algorithm 
2.
Then,  $totcost(I, {\bf t}^A) \leq B$ and
 $sepcost_B(I, {\bf t}^A) \leq \gamma \cdot  sepcost^*(I),$
where  $\gamma$ is a positive constant and $B$ is the budget calculated at line \ref{line:budget}.
}

\begin{proof}
Clearly, the Algorithm  
2 in  the  {\bf while} loop constructs a sequence ${\bf t}^A$ such that $$totcost(I, {\bf t}^A) \leq B.$$

In order to prove the second inequality in the statement of the lemma, it will be convenient to 
perform the analysis in terms of  a variant of our problem which is explicitly defined with respect
to the separation cost of a sequence of tests. 
We call this new problem the {\em Pair Separation Problem} (PSP):
The input to  the PSP,  as in the DFEP, is a 5-uple $(S, C,  {\cal X}, {\bf p}, {\bf c})$,
where $S = \{s_1, \dots, s_n\}$ is a set of objects, 
$C = \{C_1, \dots, C_m\}$ is a partition of $S$ into $m$ classes, $ {\cal X}$ is a family of subsets of $S$, ${\bf p}$ is a probability distribution on  $S,$ and ${\bf c}$ is a cost function assigning to each $X \in {\cal X}$ a cost $c(X) \in \mathbb{Q^+}.$ 
The only difference between the input of these problems is that the set of tests  $T$ in the input of DFEP is replaced with a family  ${\cal X}$ of subsets of $S$.
We say that $X \in {\cal X}$  covers an object $s$ iff  $s \in X $.
Moreover, we say that $X \in {\cal X}$  covers a pair of objects $(s,s')$ if at least one of the conditions hold: (i) $s \in X $ or (ii)
$s' \in X$.
We say that a pair  $(s,s')$ is covered by a sequence of tests if some test in the sequence covers  $(s,s')$.
The separation cost of a sequence ${\bf X}=X_1 X_2\ldots X_q$ in the instance $I_P$ of PSP is given by:

\begin{equation} 
sepcost(I, {\bf X}) = \sum_{i=1}^q p \left (X_i \setminus \bigcup_{j=1}^{i-1} X_j  \right) \left(\sum_{j=1}^i c(X_j) \right) + 
p \left ( S- \bigcup_{j=1}^q X_j \right ) \sum_{j=1}^q c(X_j).
\end{equation}

The 
{\em Pair Separation Problem}  consists of finding 
a sequence of subsets of ${\cal X}$ with minimum separation cost, $sepcost^*(I_P)$, among those sequences that cover all pairs in $S$.

An instance $I=(S, C, T, {\bf p}, {\bf c})$ of the DFEP induces an instance  $I_{P}=(S, C, {\cal X}, {\bf p}, {\bf c})$ of the PSP
where $|T|= |{\cal X}|$ and for every test $t \in T$ we have a corresponding subset $X(t) \in {\cal X}$ such that  $X(t)= \sigma_S(t)$.
Thus, in our discussion we will use the term test  $X$ to refer to a subset $X \in  {\cal X}$.
In the body of this paper we  implicitly work with the instance of the PSP
induced by  the input instance  of the DFEP. It is easy to realize that $sepcost^*(I)=sepcost^*(I_P)$.
In addition, $sepcost_B(I, {\bf t}^A)=sepcost_B(I_P, {\bf X}^A),$
where ${\bf X}^A$ is the sequence obtained from  ${\bf t}^A$  
when every $t \in  {\bf t}^A$  is replaced with $X(t)$.
Thus, in order to establish the lemma it suffices to prove that 
$$sepcost_B(I_P, {\bf X}^A) \leq \gamma \cdot  sepcost^*(I_P).$$

It is useful to observe that ${\bf X}^A$ is equal to  the sequence 
$seq$ returned by  procedure {\tt GreedyPSP} in Algorithm \ref{alg:appendix}
when it is executed on the instance $(I_P,B)$.
This algorithm corresponds to lines 4,5,9 and 10 of
the {\tt While} loop of Algorithm \ref{algo:main}. In Algorithm \ref{alg:appendix},
the greedy criteria consists of choosing the test $X$ that maximizes the ratio $p(U\cap X)/c(X).$ This is equivalent to 
the maximization of  $(p(U)-p(U\cap S^*_t))/c(t)=(p(U\cap \sigma_S(t) )/c(t)$
defining the greedy choice in Algorithm \ref{algo:main}.

\begin{algorithm}[ht]
\caption{} 
\small
{\bf Procedure} {\tt GreedyPSP} ($I_P=(S, C, {\cal X}, {\bf p})$: instance of PSP, $B$:Budget)) 
\begin{algorithmic}[1]

\STATE{  $seq \leftarrow \emptyset $, \;  $  U \leftarrow S, \; k\leftarrow 1$  }
\WHILE{ there is a test in ${\cal X}$ of cost  $\leq  B $}
  \STATE  {let $X_k$ be a test which maximizes $ \frac{p(U \cap X) }{c(X)}$ among all tests $X \in {\cal X} $  s.t.  $c(X) \leq  B $} \,\,\ (*)
\STATE { Append $X_k$ to $seq$, \; $U\leftarrow U \setminus X_k, \;  B \leftarrow B- c(X_k)$ \;, ${\cal X} \leftarrow  {\cal X} \setminus \{X_k\}$, \; $k\leftarrow k+1$}	
 \ENDWHILE
\STATE{\bf end while}
\end{algorithmic}
 \label{alg:appendix}
\end{algorithm}
\normalsize

\medskip

The proof  consists of the following steps: 
\begin{enumerate}
\item [i]  We construct an instance   $I'=(S', C', {\cal X}', {\bf p'}, {\bf c}')$ of the PSP 
from   $I_P$
\item [ii] We prove that the optimal separation cost
for $I'$ is no larger  than the optimal one for $I_P$, that is, $sepcost^*(I') \leq  sepcost^*(I_P)$. 
\item [iii]
we prove that separation cost $sepcost(I',{\bf X'})$ of any sequence of tests ${\bf X'}$ returned by the above pseudo-code on the instance $(I',B)$ 
is at a constant factor of $sepcost^*(I')$, that is, $sepcost(I',{\bf X'})$ is $O(sepcost^*(I'))$.
\item[iv]  we prove that there exists a sequence of tests ${\bf Z}$ possibly returned by {\tt GreedyPSP} when 
executed on the instance $(I', B)$ such that  
$sepcost_B(I_P,{\bf X}^A) \leq 2 sepcost(I',{\bf Z})$.
\end{enumerate}

By chaining these inequalities, we conclude that  
$sepcost_B(I_P,{\bf X}^A) $ is $O(sepcost^*(I_P)).$
The steps (ii), (iii) and (iv) are proved in Claims 1,2 and 3, respectively.
We start with the construction of instance $I'$.

\medskip

\noindent {\it Construction of instance $I'$.}
For every test $X \in {\cal X}$,  we define $n(X)= 2c(X)$.

The instance $I' = (S', C', {\cal X}', {\bf p'}, {\bf c}')$ is constructed from $I_{P} = (S, C, {\cal X}, {\bf p}, {\bf c})$ as follows.
Let $N=\prod_{X \in {\cal X}} n(X)$.
For each  $s \in S $ we add   $N$ objects to $S'$,
each of them with probability $p(s)/N$ and with class equal to that of  $s$.
If an object $s'$ is added to set $S'$ due to $s$, we say that
$s'$ is generated from $s$. 

For every test $X \in {\cal X}$ we add  
$n(X)$ tests to the set ${\cal X}'$, each of them with cost $1/2$.   
If a test $X'$ is added to set ${\cal X}'$ due to $X$, we say that
$X'$ is generated from $X$.

It remains to define to which subset of $S'$ each test $X' \in {\cal X}'$ corresponds to.
If $s \notin X$ then 
$s' \notin X' $ for every $s'$ generated from $s$ and
every $X'$ generated from $X$.
Let ${\cal X}_s=\{X^1,\ldots,X^{|{\cal X}_s|} \}$ be the set of tests that contains the object $s \in S$.
Note that the number of tuples $(\theta^1,\ldots,\theta^{|{\cal X}_s|} )$, where $ \theta^i \in {\cal X}' $ is
a test generated from $X^i \in {\cal X}_s$ is $\prod_{X \in {\cal X}_s} n(X)$.
 Thus, we
create a one to one correspondence between these tuples and the  numbers in the set $Poss(s)=\{1,\ldots, \prod_{X \in {\cal X}_s} n(X)\}$.
For a test $\theta \in {\cal X}'$, generated from $X \in {\cal X}_s$, let $F(\theta) \subset Poss(s)$
be the set of numbers
that correspond to the tuples that includes $\theta$.
Note that 
\begin{equation}
|F(\theta)|=\left ( \prod_{Y \in {\cal X}_s} n(Y) \right ) / n(X).
\end{equation}
In addition, we associate each object $s'\in S'$, 
generated from $s$, with a number $f(s') \in Poss(s)$ in a balanced way
so that each number in $Poss(s)$ is associated with $N/\prod_{X \in {\cal X}_s} n(X)$ objects.
Thus, a test $\theta \in {\cal X}'$, generated from $X \in {\cal X}_s$, covers an object $s'$ generated from $s$ if and only if $f(s') \in F(\theta)$.

For the instance $I'$ we have the following useful properties: 
\begin{itemize} 

\item[a] if  $X\in {\cal X}$ covers object  $s \in S$ then  each test $\theta\in {\cal X}'$, generated from $X$, covers
exactly $N/n(X)$ objects generated from $s$. Moreover,
each object generated from $s$ is covered by exactly one test generated 
from $X$.
 
\remove{
\item[b] For some $j \ge 1$, let $X^1,X^2,\ldots,X^j$ be a sequence of tests in ${\cal X}$ (the set of tests from instance $I_P$). In addition,
let $seq$ be a sequence of tests in ${\cal X}'$ (the set of tests from instance $I'$), consisting
of  $n(X^1)$ distinct tests generated from $X^1$, followed by   $n(X^2)$  distinct tests generated
from $X^2$ and so on, ending with a sequence of $n(X^j)$ distinct tests generated
from $X^j$. Let $\theta^i_h$ be the $h$-th test generated from
$X^i$ in $seq$.
Then, for $1 \leq i \leq j$ and
$1 \leq h \leq n(X^i)$, we have that

\begin{equation} \frac{p(X^i - ( X^1 \cup \ldots \cup X^{i-1}))}{c(X^i)} = \frac{p(\theta^i_h - seq^i_{h-1})}{c(\theta^i_h)},
\end{equation}
where $seq^i_{h-1}$ is the set of objects  of the sequence of tests that precedes 
$\theta^i_h$ in  sequence $seq$.
This property will be useful
to compare, in Claim 3, the sequences of tests 
generated by Algorithm \ref{alg:appendix}  when it is
executed over instances $(I_P,B)$ and $(I',B)$.
}



\item[b] 
If a set of tests $G' \subseteq {\cal X}'$  covers all pairs of $I'$ then
the set $G=\{ X \in {\cal X}|$ all tests generated from $X$ belong to $G' \}$
covers all pairs of $I_P$. 
\end{itemize}

\noindent
Property (a) holds because a test $\theta$ generated by $X$ is 
associated with $|F(\theta)|=\prod_{Y \in {\cal X}_s} n(Y) / n(X)$ numbers in $Poss(s)$ and
to each number in $Poss(s)$ we have $N/\prod_{Y \in {\cal X}_s} n(Y)$ objects associated with.

\remove{
To see that property (b) holds, let $\pi_i=X^i-(X^1 \cup \ldots \cup X^{i-1})$. 
Note that if an object 
$s$ belongs to 
$\pi_i$ then, by property (a), each of the $n(X^i)$ tests generated from 
$X^i$ covers exactly $N/n(X^i)$ objects generated from $s$.
Thus, the sum of the probabilities of the objects
covered by  $\theta^i_h$ that
are not in $seq^i_{h-1}$, namely  $p(\theta^i_h - seq^i_{h-1})$, satisfies 
$$ p(\theta^i_h - seq^i_{h-1}) = \sum_{s \in \pi_i} \frac{p(s)}{N} \frac{N}{n(X^i)} = \frac{p(\pi_i)}{ n(X^i)}.$$
Thus, 

$$\frac{ p(\theta^i_h - seq^i_{h-1})}{c(\theta^i_h)}= \frac{ 2p(\pi_i) }{ n(X^i)} =\frac{ p(\pi_i) }{ c(X^i)}=\frac{p(X^i - ( X^1 \cup \ldots \cup X^{i-1}))}{c(X^i)}$$
}


To see that property (b) holds, let us  assume that $G'$ covers all pairs of the instance $I'$ and 
$G$ does not cover a pair $(s_1,s_2)$. Let  ${\cal X}_{s_1}=\{X_1^1,\ldots,X_1^x\}$ and
${\cal X}_{s_2}=\{X_2^1,\ldots,X_2^y\}$ be  
the set of tests that covers  $s_1$ and $s_2$,  respectively. 
The fact that $G$ does not cover  $(s_1,s_2)$ implies that $ ({\cal X}_{s_1} \cup {\cal X}_{s_2}) \cap G = \emptyset$ so that
 for each $X_1^i \in {\cal X}_{s_1}$, there is a test $\theta_1^{i}$, generated from $X_1^i$,
that does not belong to $G'$.  Similarly, 
for each $X_2^i \in {\cal X}_{s_2}$, there is a test $\theta_2^{i}$, generated from $X_2^i$,
that does not belong to $G'$.
Let $s'_1$ be an object, generated from $s_1$, that is mapped, via function $f(\cdot)$, into the number in $Poss(s_1)$ that corresponds to the tuple
$(\theta_1^{1},\ldots,\theta_1^{x})$. Moreover, let 
$s'_2$ be an object, generated from $s_2$,
that is mapped, via function $f(\cdot)$, into the number in $Poss(s_2)$ that corresponds to the tuple
$(\theta_2^{1},\ldots,\theta_2^{y})$. 
The pair $(s'_1,s'_2)$ is not covered by $G'$, which is a contradiction.

\medskip

\noindent {\bf Claim 1.}
The optimal separation cost
for $I'$ is no larger than the optimal separation cost for $I_{P},$ i.e., $sepcost^*(I') \leq sepcost^*(I_{P}).$

Given a  sequence ${\bf X}_P$ for 
$I_P$ that covers all $P(S)$ pairs we can obtain 
a  sequence ${\bf X}$ for 
$I'$ by replacing each test $X_P \in  {\bf X}_P$ with
the  $n(X_P)=2 c(X_P)$ tests in ${\cal X}'$ that were generated from 
$X_P$, each of which has cost $1/2.$
It is easy to see that ${\bf X}$ covers all the pairs in $I'$ and  the separation cost of
${\bf X}$ is not larger than that of $ {\bf X}_P$.
This establishes our claim. \mbox

\medskip

Now let
${\bf X'}$ be a sequence of tests  returned by procedure {\tt GreedyPSP} in Algorithm \ref{alg:appendix}  
 when it is executed on the instance $(I',B)$.

\medskip

\noindent {\bf Claim 2.}
The separation cost of the sequence ${\bf X'}$ is at most a constant factor
of that of ${\bf X}^* = X_1^*, \dots, X_{q^*}^*$, which is the  sequence of tests
with minimum separation cost among all sequences of tests covering all the pairs, for the instance 
$I',$ i.e., $sepcost(I', {\bf X'}) \leq \beta  sepcost^*(I'),$ for some constant $\beta.$

Let $p_j$ (resp. $p^*_j$) be the sum of the probabilities of the objects covered by the first $j$ tests in ${\bf X'}$ (resp.\ ${\bf X}^*$). 
In particular, we have $p_0 = p_0^* = 0.$
In addition, let $Q$ be the sum of the  probabilities of all objects in $S'$.
Notice that, with the above notation, we can rewrite the separation cost of the sequence ${\bf X'} = X'_1, \dots, X'_q$ as  
$$sepcost(I', {\bf X'}) = \sum_{j=1}^q c(X'_j) (Q - p_{j-1})=\sum_{j=1}^q (1/2) \cdot (Q - p_{j-1}).$$

Let $\ell$ be such that $ 2^{\ell-1} \leq B \leq 2^{\ell}-1$, where $B$ is the budget in the statement of the lemma.
For $j=0,\ldots, \ell$, let $i_j=2^{j+2}-2 $ and $i^*_j=2^{j+1} $.
Furthermore, let $P^{[j]}$ be the  sum of the probabilities of the objects covered by the first $i_j$
tests of ${\bf X'}$. In formulae,
$P^{[j]} = p \left (\bigcup_{k=1}^{i_j} X'_k \right).$  
Analogously, let $P^{[j]}_*$ be the sum of the probabilities  of the objects covered by the first $i^*_j$ tests in ${\bf X}^*.$  In formulae,  
$P^{[j]}_* = p\left(\bigcup_{k=1}^{i^*_j} X^*_k \right).$ 
For the sake of definiteness, we set $i_{-1} = i^*_{-1} = 0$ and $P^{[-1]} =P^{[-1]}_*=0$

Then, we have

\begin{eqnarray*} 
sepcost(I', {\bf X'}) &=& \sum_{i=1}^q (1/2) \cdot (Q - p_{i-1}) 
\leq  \sum_{j=0}^{\ell-1} \sum_{k = i_{j-1}+1}^{i_j} (1/2) \cdot (Q-p_{k-1}) \\
&\leq& \sum_{j=0}^{\ell-1}  \sum_{k = i_{j-1}+1}^{i_j} (1/2)(Q-P^{[j-1]})  
\leq \sum_{j=0}^{\ell-1} 2^{j} (Q-P^{[j-1]}) \\
&\leq& Q+ \sum_{j=1}^{\ell-1} 2^{j} (Q-P^{[j-1]}) 
\leq Q+ \sum_{j=0}^{\ell-2} 2^{j+1} (Q-P^{[j]}),
\end{eqnarray*}
where the first inequality holds because $q \leq 2B \leq 2^{\ell+1} -2 = i_{\ell-1}$ and
the second one holds because $p_{k-1} \geq P^{[j-1]}= p_{i_{j-1}} $ for $k \geq i_{j-1}+1$.


We now devise a lower bound on 
the separation cost of ${\bf X}^*.$ For this, we 
first  note that the length $q^*$ of  ${\bf X}^*$ is at least $2B \geq 2^{\ell}= i^*_{\ell-1}$,
for otherwise the property (b) of instance $I'$  would guarantee the existence  of a sequence of tests of total cost smaller than $B$ that
covers all pairs for instance $I_P$ (and for the instance $I$ of the DFEP as well), which contradicts Lemma \ref{lemma:length}. Therefore, we can lower bound the
the separation cost of the sequence ${\bf X}^*$ as follows:  

\begin{eqnarray} 
sepcost(I', {\bf X}^*) &=& \sum_{i=1}^{q^*} c(X_i^*) (Q - p^*_{i-1}) \label{eq:27may0} \\
&\geq& \sum_{ i=1}^{i^*_0} (1/2) \cdot (Q - p^*_{i-1})  + \sum_{j=1}^{\ell-1} \sum_{k = i^*_{j-1}+1}^{i_j^*} (1/2) \cdot (Q-p^*_{k-1}) \label{eq:27may1}\\
&\geq&  \frac{2Q- p^*_1 }{2}+ \sum_{j=1}^{\ell-1} (2^{j} - 2^{j-1}) (Q-P_*^{[j]}) \geq  \frac{3Q}{4} + \frac{1}{2} \sum_{j=1}^{\ell-1} 2^{j} (Q-P_*^{[j]}) 
 \label{eq:27may2}
\end{eqnarray}


The inequality in (\ref{eq:27may1}) follows from (\ref{eq:27may0}) by considering in the summation on the 
right hand side of (\ref{eq:27may1})  only 
the first $i^*_{\ell -1} = 2^{\ell} \leq B \leq q^*$ tests.

The term  $(2Q- p^*_1)/2$ in the first inequality (\ref{eq:27may2}) is  the contribution
of the first two tests of the sequence  ${\bf X}^*$ to the separation cost.
To prove that $\frac{2Q- p^*_1 }{2} \geq  \frac{3Q}{4}$ in the last inequality, we note that 
that $p^*_1 \leq Q/2$ because the probability
covered by the first  test $X^*_1$ of sequence ${\bf X^*}$ is
$p(X)/n(X) \leq  Q/2c(X) \leq Q/2$, where $X$ is the test that generates $X^*_1$.
In the last inequality  we used the fact that $c(X) \geq 1$ for all $X \in {\cal X}$.

Let $S'_{k} \subseteq S'$ be the set of objects covered by the sequence of tests $X'_1,X'_2,\ldots,X'_{k}$, which
 is the prefix of length $k$ of the sequence of tests ${\bf X'}$.
We shall note that 
for $l \geq k+1$, the subsequence $X'_{k+1},\ldots,X'_l$ of ${\bf X'}$ coincides with the sequence of tests 
constructed   through the execution of {\tt Adapted-Greedy} over the instance ($S' \setminus S'_{k},\tilde{T}, f_2, {\bf c}', B'$), where 
\begin{itemize}
\item $\tilde{T} ={\cal X}'  \setminus \{X'_1,X'_2,\ldots,X'_k\}$ is a set of tests,  all of them with cost $1/2;$
\item the function $f_2$ maps a set of tests into 
the probability of the objects in $S' \setminus S'_{k}$ that are covered by the tests in the set;
\item $B' = \frac{(l-k)}{2}.$
\end{itemize}

Since the set $\{X^*_{1},X^*_{2},\ldots,X^*_{l-k}\} -  \{X'_1,X'_2,\ldots,X'_{k}\}$ is 
 a feasible  solution
for this instance, it follows from Theorem 2 
that $p_l - p_k \geq \hat{\alpha} (p^*_{l-k} - p_k)$,
where $\hat{\alpha} = 1 - \frac{1}{e}.$ By setting $l=i_j$ and $k=i_{j-1}$ we get that
 $$P^{[j]} - P^{[j-1]} \geq \hat{\alpha} (P_*^{[j-1]} - P^{[j-1]}).$$

It follows that $$Q - P^{[j]} \leq \hat{\alpha}(Q - P_*^{[j-1]}) + (1-\hat{\alpha}) (Q - P^{[j-1]}).$$ 
Thus, setting $$U = Q+ \sum_{j=0}^{\ell-2} 2^{j+1} (Q-P^{[j]}),$$ which is the upper bound we derived on the separation 
cost of the sequence ${\bf X'}$,  we have 

\begin{eqnarray*}
U &=&    Q + \sum_{j=0}^{\ell-2} 2^{j+1} (Q-P^{[j]})\\
&\leq&  Q + \hat{\alpha} \sum_{j=0}^{\ell-2} 2^{j+1} (Q-P_*^{[j-1]}) + (1-\hat{\alpha}) \sum_{j=0}^{\ell-2} 2^{j+1} (Q-P^{[j-1]})\\
&=&  Q + 2\hat{\alpha}Q + \hat{\alpha}\sum_{j=1}^{\ell-2} 2^{j+1} (Q-P_*^{[j-1]}) + 2(1-\hat{\alpha})Q +  (1-\hat{\alpha}) \sum_{j=1}^{\ell-2} 2^{j+1} (Q-P^{[j-1]})\\
&=&  Q + 2 Q +  2\hat{\alpha}  \sum_{j=0}^{\ell-3} 2^{j+1} (Q - P_*^{[j]}) +2 (1-\hat{\alpha})  \sum_{j=0}^{\ell-3} 2^{j+1} (Q-P^{[j]})\\
&\leq& Q + 2 Q + 4 \hat{\alpha} Q  +  2\hat{\alpha}  \sum_{j=1}^{\ell-3} 2^{j+1} (Q - P_*^{[j]}) +2 (1-\hat{\alpha})  \sum_{j=0}^{\ell-3} 2^{j+1} (Q-P^{[j]})\\
&= & (1-2(1-\hat{\alpha})+2+4\hat{\alpha}) Q +   4\hat{\alpha}  \sum_{j=1}^{\ell-3} 2^{j} (Q - P_*^{[j]}) 
+ 2(1-\hat{\alpha})\left(Q + \sum_{j=0}^{\ell-3} 2^{j+1} (Q-P^{[j]})\right)\\
& \leq & (1+6 \hat{\alpha}) Q +   4\hat{\alpha}  \sum_{j=1}^{\ell-1} 2^{j} (Q - P_*^{[j]}) 
+ 2(1-\hat{\alpha})U\\
&\leq& (8\hat{\alpha}+4/3)  \, sepcost(I', {\bf X}^*) +  2(1-\hat{\alpha}) U,
\end{eqnarray*}
where the last inequality follows from equation \ref{eq:27may2}.
Thus,  we obtain $$sepcost(I', {\bf X'}) \leq U \leq  \frac{(8\hat{\alpha}+4/3)}{2\hat{\alpha}-1}  sepcost(I', {\bf X}^*).$$

\medskip

For the last claim let ${\bf X}^A$ be the sequence obtained by {\tt GreedyPSP} (Algorithm \ref{alg:appendix})
when it is executed on instance $(I_P,B)$.

\medskip

%

\noindent {\bf Claim 3.}  There exists an execution of procedure {\tt GreedyPSP} (Algorithm \ref{alg:appendix})
on instance $(I',B)$ which returns a sequence ${\bf Z}$ satisfying
$sepcost_B(I_P,{\bf X}^A) \leq 2 sepcost(I',{\bf Z})$.


Let $X^A_i$ be the $i$-th test of sequence ${\bf X}^A$ 
and let $X^A_{r+1}$ be the first test of ${\bf X^A}$
that is  not the  test which maximizes 
$p(U\cap X) /c(X) $ among all the tests in ${\cal X}$ in line (*)
of Algorithm \ref{alg:appendix}.  Note that $X^A_{r+1}$ is chosen
by Algorithm \ref{alg:appendix} rather than $Y$, the test which maximizes $p(U\cap X) /c(X)$,
because  $Y$  has cost larger than remaining budget $z = B - \sum_{j=1}^r c(X_j^A)$.
The case where $X^A_{r+1}$ does not
exist is easier to handle and will be discussed at the end of the proof.
Because $X^A_1,\ldots,X^A_{r}$ is a prefix of ${\bf X}^A$ we have  $$sepcost_B(I_P,{\bf X}^A) \leq  
sepcost_B(I_P, \langle X^A_1,\ldots,X^A_{r}\rangle).$$
Thus, to establish the claim it suffices to show that   
$$sepcost_B(I_P, \langle X^A_1,\ldots,X^A_{r} \rangle) \leq  2 sepcost(I',{\bf Z}),$$ where ${\bf Z}$ is  
a possible output of {\tt GreedyPSP} (Algorithm \ref{alg:appendix})
on instance $(I',B).$ 

For $j = 1, \dots, r$, let ${\bf Z}^{[j]} =  \langle Z^{[j]}_1, \dots, Z^{[j]}_{n(X_j^A)}\rangle$ be a sequence of tests defined by some permutation of the $n(X^A_j)$ tests in ${\cal X}'$, generated from $X_j^A.$ 

Let $z = B - \sum_{j=1}^r c(X_j^A)$ and ${\bf Z}^{[r+1]} =  \langle Z^{[r+1]}_1, \dots, Z^{[r+1]}_{2z}\rangle$ 
be a sequence of $2 z$ of the $n(Y) = 2c(Y) > 2z$ tests in ${\cal X}'$, generated from $Y.$ 
The proof of the following proposition is deferred to section \ref{app:prop}.

\begin{proposition} \label{prop:claim3-apendix}
Let ${\bf Z} =  {\bf Z}^{[1]} \, {\bf Z}^{[2]}\, \dots \, {\bf Z}^{[r]} \, {\bf Z}^{[r+1]}$ be
the sequence obtained by the juxtaposition of the sequences ${\bf Z}^{[1]}, \dots, {\bf Z}^{[r+1]}.$ Then, for  
${\bf Z}$ the following conditions hold:
\begin{itemize}
\item[(i)] for each $j=1, \dots, r+1$ and $\kappa \neq \kappa' \in \{1, \dots, n(X^A_j) = 2c(X^A_j)\}$, 
it holds that $$Z^{[j]}_{\kappa} \cap Z^{[j]}_{\kappa'} = \emptyset$$

\item[(ii)] For each $j=1, \dots, r+1$ and each test  $Q$  in ${\cal X}',$ with $X$ being the test in ${\cal X}$ from which 
 $Q$ is generated, it holds that
$$
\frac{p(Q - \bigcup_{i =1}^{j-1} \bigcup_{\kappa'=1}^{n(X^A_{i})} Z^{[i]}_{\kappa'} )}{c(Q)} 
= \frac{p(X - \bigcup_{i=1}^{j-1} X^A_{i})}{c(X)}
$$

\item[(iii)] ${\bf Z}$ is a feasible output for  {\tt GreedyPSP} (Algorithm \ref{alg:appendix}) on instance $(I',B).$
\end{itemize}
\end{proposition}


First note that the sequence ${\bf Z}$ has length $2B$ and total cost $B$. This is easily verified by 
recalling that: (i) each test in the sequence ${\bf Z}$ has cost $1/2$; 
(ii) for each $j=1, \dots, r,$ the subsequence ${\bf Z}^{[j]}$ has length $n(X^A_j),$ hence 
$totcost(I', {\bf Z}^{[j]}) =  n(X^A_j)/2 = c(X^A_j)$; 
(iii) the subsequence ${\bf Z}^{[r+1]}$ has length 
$2z = 2B - 2\sum_{j=1}^r c(X_j^A) = 2(B - \sum_{j=1}^r |{\bf Z}^{[j]}|)$ hence 
$totcost(I', {\bf Z}^{[r+1]}) = z = B - \sum_{j=1}^r totcost(I', {\bf Z}^{[j]})$. 


Let $C_0=0$, and for $j=1, \dots, r,$ let  $C_j= \sum_{i=1}^j c(X^A_i)$. By the observations in the previous paragraph, we also 
have $C_j = totcost(I', <{\bf Z}^{[1]}\,\dots \, {\bf Z}^{[j]}>) = \sum_{i=1}^j \sum_{\kappa = 1}^{n(X^A_i)} c(Z_{\kappa}^{[i]}).$ 

By grouping objects which incur the same cost in ${\bf X},$ 
we can write $sepcost_B(I_P, X^A_1,\ldots,X^A_r)$ as follows
\begin{equation} \label{eq:sep_B-X}
 sepcost_B(I_P, \langle X^A_1,\ldots,X^A_r \rangle) = 
\sum_{j=1}^r C_j \cdot p\left(X^A_j - \bigcup_{i=1}^{j-1} X^A_i\right) + B \cdot p\left(S - \bigcup_{j=1}^{r} X^A_j\right)
\end{equation}

Analogously, we can compute $sepcost(I',{\bf Z})$ as follows:
\begin{eqnarray}
sepcost(I',{\bf Z}) &=& \nonumber
\sum_{j=1}^r \sum_{\kappa=1}^{n(X_j^A)} 
p\left(Z_{\kappa}^{[j]} - \left(\bigcup_{i=1}^{j-1} \bigcup_{\kappa'=1}^{n(X^A_i)} Z^{[i]}_{\kappa'} \right) 
- \left( \bigcup_{\kappa'=1}^{\kappa-1} Z^{[j]}_{\kappa'} \right)\right) \left ( C_{j-1} + \kappa \cdot \frac{1}{2} \right )\\
& &~~~~+ \sum_{\kappa = 1}^{2z} 
p\left(Z_{\kappa}^{[r+1]} - \left(\bigcup_{i=1}^{r} \bigcup_{\kappa'=1}^{n(X^A_i)} Z^{[i]}_{\kappa'} \right) 
- \left( \bigcup_{\kappa'=1}^{\kappa-1} Z^{[r+1]}_{\kappa'} \right)\right) \left ( C_{r} + \kappa \cdot \frac{1}{2} \right )
\nonumber
\\
& &~~~~+ B \cdot 
p\left(S' - \left(\bigcup_{i=1}^{r} \bigcup_{\kappa'=1}^{n(X^A_i)} Z^{[i]}_{\kappa'} \right) 
- \left( \bigcup_{\kappa=1}^{2z} Z^{[r+1]}_{\kappa} \right)\right), \label{eq:24}
\end{eqnarray}
where we have splitted ${\bf Z}$ into the objects covered by the subsequences 
${\bf Z}^{[1]}, \dots, {\bf Z}^{[r]},$ the objects covered by the subsequence ${\bf Z}^{[r+1]}$ 
 and the remaining objects. 
In the above expressions, the term 
$Z_{\kappa}^{[j]} - \left(\bigcup_{i=1}^{j-1} \bigcup_{\kappa'=1}^{n(X^A_i)} Z^{[i]}_{\kappa'} \right) 
- \left( \bigcup_{\kappa'=1}^{\kappa-1} Z^{[j]}_{\kappa'} \right)$
represents the set of objects covered by $Z_{\kappa}^{[j]}$ and not covered by any of the preceding tests in 
${\bf Z}.$ The cost of separating each of these objects is the sum of the costs of all the 
tests performed up to $Z_{\kappa}^{[j]}$, i.e.,  
$\sum_{i=1}^{j-1} totcost(I', {\bf Z}^{[i]}) + \kappa/2 = C_{j-1} + \kappa/2.$
 
Now, we notice that for each 
$j=1, \dots, r+1$ and $1 \leq \kappa \leq \min\{n(X^A_j), 2z\}$ the set of objects covered by
$Z^{[j]}_{\kappa}$ and not covered by the previous tests are
$$ Z_{\kappa}^{[j]} - \left(\bigcup_{i=1}^{j-1} \bigcup_{\kappa'=1}^{n(X^A_i)} Z^{[i]}_{\kappa'} \right) 
- \left( \bigcup_{\kappa'=1}^{\kappa-1} Z^{[j]}_{\kappa'} \right) = 
Z_{\kappa}^{[j]} - \left(\bigcup_{i=1}^{j-1} \bigcup_{\kappa'=1}^{n(X^A_i)} Z^{[i]}_{\kappa'} \right)$$
where the equality follows because by Proposition \ref{prop:claim3-apendix} (i) we have that $Z_{\kappa}^{[j]} \cap \left( \bigcup_{\kappa'=1}^{\kappa-1} Z^{[j]}_{\kappa'} \right) = \emptyset.$

Moreover, by Proposition \ref{prop:claim3-apendix} (ii) we have that 
$$p\left(Z_{\kappa}^{[j]} - \left(\bigcup_{i=1}^{j-1} \bigcup_{\kappa'=1}^{n(X^A_i)} Z^{[i]}_{\kappa'} \right)\right) = 
\frac{p\left(X^A_j - \bigcup_{i=1}^{j-1} X^A_i \right)}{n(X^A_j)}.$$ 
Finally, the set
$$ R= S' - \left(\bigcup_{i=1}^{r} \bigcup_{\kappa'=1}^{n(X^A_i)} Z^{[i]}_{\kappa'} \right) 
- \left( \bigcup_{\kappa=1}^{2z} Z^{[r+1]}_{\kappa} \right) $$
that appears  in the third term of the righthand side of (\ref{eq:24})
can be spilt into 
the objects covered by the tests generated from $Y$  and the remaining ones.
By using arguments similar to those employed above one can realize that 
the objects covered by the tests generated from $Y$ 
are  exactly those  generated  by the objects  in  $Y - \bigcup_{i=1}^{r} X^A_i$
that are not covered by the tests in $\left( \bigcup_{\kappa=1}^{2z} Z^{[r+1]}_{\kappa} \right)$.
Thus, their contribution to $p(R)$ is
$\left ( \frac{n(Y)-2z}{n(Y)} \right ) p(Y - \bigcup_{i=1}^{r} X^A_i) $.
On the other hand, 
the contribution to $p(R)$ of the remaining  objects  is $p(S - Y - \bigcup_{j=1}^r X^A_j).$

Therefore, we can rewrite (\ref{eq:24}) as follows

\begin{eqnarray}
sepcost(I',{\bf Z}) &=& 
\sum_{j=1}^r \sum_{\kappa=1}^{n(X_j^A)} 
\frac{p\left(X^A_j - \bigcup_{i=1}^{j-1} X^A_i \right)}{n(X^A_j)} \left ( C_{j-1} + \kappa \cdot \frac{1}{2} \right )
\nonumber
\\
& &~~~~+ \sum_{\kappa = 1}^{2z} 
\frac{p\left(Y - \bigcup_{i=1}^{r} X^A_i \right)}{n(Y)} \left ( C_{r} + \kappa \cdot \frac{1}{2} \right )
\nonumber
\\
& &~~~~+ \frac{n(Y)-2z}{n(Y)} p \left ( Y - \bigcup_{i=1}^{r} X^A_i \right ) \cdot B + 
p \left (S - Y - \bigcup_{j=1}^r X^A_j \right ) B \label{eq:25}
\end{eqnarray}

Via simple algebraic manipulation on the first term in the right hand side of (\ref{eq:25}) we have that
$$\sum_{\kappa=1}^{n(X_j^A)} \frac{1}{n(X^A_j)} \left ( C_{j-1} + \kappa \cdot \frac{1}{2} \right )
= C_{j-1} + \frac{n(X_j^A) + 1}{4} \geq C_{j-1} + \frac{c(X_j^A)}{2},$$
and, analogously, for the second term in the right hand side of (\ref{eq:25}) we have
$$\frac{1}{n(Y)} \sum_{\kappa = 1}^{2z}  \left ( C_{r} + \kappa \cdot \frac{1}{2} \right ) =
\frac{2z}{n(Y)} \left(C_r + \frac{2z+1}{4}\right) \geq \frac{2z}{n(Y)} \frac{B + C_r}{2}.$$

Hence, we have 

\begin{eqnarray}
sepcost(I',{\bf Z}) &\geq& 
\sum_{j=1}^r  p\left(X^A_j - \bigcup_{i=1}^{j-1} X^A_i \right)
 \left ( C_{j-1} + \frac{c(X_j^A)}{2} \right )
+ p\left(Y - \bigcup_{i=1}^{r} X^A_i \right) \frac{2z}{n(Y)} \frac{B + C_r}{2}
\nonumber
\\
& &~~~~+ \frac{n(Y)-2z}{n(Y)} p \left ( Y - \bigcup_{i=1}^{r} X^A_i \right ) \cdot B + 
p \left ( S - Y - \bigcup_{j=1}^r X^A_j \right ) B \label{eq:26}
\end{eqnarray}

Finally, we observe that $C_{j-1} + \frac{c(X_j^A)}{2} \geq C_j/2$ and $(B+C_r)/2 \geq B/2.$ 
Then, the sum of the second and third term in the right hand side of (\ref{eq:26}) can be lower bounded 
with $p(Y-\bigcup_{j=1}^r X^A_j) \cdot B/2$ and we get

\begin{equation}
sepcost(I',{\bf Z}) \geq 
\sum_{j=1}^r  p\left(X^A_j - \bigcup_{i=1}^{j-1} X^A_i \right) 
\frac{C_{j}}{2} 
+ p\left(Y - \bigcup_{i=1}^{r} X^A_i \right) \frac{B}{2} +
p(S - Y - \bigcup_{j=1}^r X^A_j) B \label{eq:27}
\end{equation}

Putting together (\ref{eq:27}) and (\ref{eq:sep_B-X}) we have the desired result

$$sepcost(I',{\bf Z}) \geq 2 sepcost_B(I_P, <X^A_1, \dots, X^A_r>).$$

\bigskip

It remains to argue about the case where  $X^A_{r+1}$ does not exist,
which means that all tests that maximize the greedy criteria in Algorithm
\ref{alg:appendix} have cost smaller than the current budget $B$.
In this case, the analysis becomes simpler and  be easily handled in the same way as above. 
In fact, the only difference is that 
the last term in (\ref{eq:sep_B-X}) disappears, as well as all the terms referring to $Y$ and ${\bf Z}^{[r+1]}.$ 

The lemma follows from the correctness of the three claims.
\end{proof}

\section{The Proof of Proposition \ref{prop:claim3-apendix}} \label{app:prop}

\noindent
{\bf Proposition \ref{prop:claim3-apendix}.}
{\em Let ${\bf Z} =  {\bf Z}^{[1]} \, {\bf Z}^{[2]}\, \dots \, {\bf Z}^{[r]} \, {\bf Z}^{[r+1]}$ be
the sequence obtained by the juxtaposition of the sequences ${\bf Z}^{[1]}, \dots, {\bf Z}^{[r+1]}.$ Then, for  
${\bf Z}$ the following conditions hold:
\begin{itemize}
\item[(i)] for each $j=1, \dots, r+1$ and $\kappa \neq \kappa' \in \{1, \dots, n(X^A_j) = 2c(X^A_j)\}$, 
it holds that $$Z^{[j]}_{\kappa} \cap Z^{[j]}_{\kappa'} = \emptyset$$

\item[(ii)] For each $j=1, \dots, r+1$ and each test $Q$  in ${\cal X}',$ with $X$ being the test in ${\cal X}$ from which 
 $Q$ is generated, it holds that
$$
\frac{p(Q - \bigcup_{i =1}^{j-1} \bigcup_{\kappa'=1}^{n(X^A_{i})} Z^{[i]}_{\kappa'} )}{c(Q)} 
= \frac{p(X - \bigcup_{i=1}^{j-1} X^A_{i})}{c(X)}
$$

\item[(iii)] ${\bf Z}$ is a feasible output for  {\tt GreedyPSP} (Algorithm \ref{alg:appendix}) on instance $(I',B).$
\end{itemize}
}
\begin{proof}

\noindent
Item (i) is a direct consequence of property (a) of the instance $I'$.

\medskip

\noindent
In order to prove (ii), we observe that, from the definition of the sequences ${\bf Z}^{[i]}$ ($i=1, \dots, r$), it follows
that the elements of  $W = \bigcup_{i=1}^{j-1} \bigcup_{\kappa'=1}^{n(X^A_i)} Z^{[i]}_{\kappa'}$ 
are all the elements in $S'$ which are generated from $\bigcup_{i=1}^{j-1} X^A_i.$
Therefore, the elements of $Q - W$ are precisely the elements of $Q$ which are 
generated from $X - \bigcup_{i=1}^{j-1} X^A_i.$ 

For each $s \in X - \bigcup_{i=1}^{j-1} X^A_i,$  there are
precisely $\frac{N}{n(X)}$ elements  in $Q$ that are generated from $s$,  and each one of them 
has probability $p(s)/N.$ Hence we have 
$$\frac{p \left (Q - \bigcup_{i =1}^{j-1} \bigcup_{\kappa'=1}^{n(X^A_{i})} Z^{[i]}_{\kappa'} \right )}{c(Q)} = \frac{1}{c(Q)}\sum_{s \in X - \bigcup_{i=1}^{j-1} X^A_i} \frac{N}{n(X)}\frac{p(s)}{N} =
\frac{1}{c(Q) n(X)} \sum_{s \in X - \bigcup_{i=1}^{j-1} X^A_i} p(s),$$
from which we have (ii), since $1/c(Q)n(X) = 2/n(X) = c(X).$

\medskip

\noindent
In order to prove (iii) it is enough to show that the following claim holds.

\noindent
{\em Claim.} For each $j = 1, \dots, r+1$ and 
$\kappa = 1, \dots, \min\{2z, n(X^A_j)\}$ we have that 
\begin{equation} \label{eq:Z-greedy}
\frac{p\left(Z_{\kappa}^{[j]} - \left(\bigcup_{i=1}^{j-1} \bigcup_{\kappa'=1}^{n(X^A_i)} Z^{[i]}_{\kappa'} \right) 
- \left( \bigcup_{\kappa'=1}^{\kappa-1} Z^{[j]}_{\kappa'} \right)\right)}{c(Z_{\kappa}^{[j]})} 
\geq 
\frac{p\left(Q - \left(\bigcup_{i=1}^{j-1} \bigcup_{\kappa'=1}^{n(X^A_i)} Z^{[i]}_{\kappa'} \right) 
- \left( \bigcup_{\kappa'=1}^{\kappa-1} Z^{[j]}_{\kappa'} \right)\right)}{c(Q)} 
\end{equation}
for any $Q \in {\cal X}'.$

\medskip

This claim  says that, for each $j=1, \dots, r+1$ and $1 \leq \kappa \leq \min\{2z, n(X^A_j)\},$ 
 if ${\bf Z}$ has been constructed up to the test preceding $Z_{\kappa}^{[j]}$ then with respect to the tests already
chosen, the test  $Z_{\kappa}^{[j]}$ satisfies the greedy criterium of procedure {\tt GreedyPSP}. This implies that 
${\bf Z}$ is a feasible output for {\tt GreedyPSP}, as desired.

\bigskip

\noindent
{\em Proof of the Claim.}
Let $R$  be the quantity on the right hand side of (\ref{eq:Z-greedy}), and $X$ be the test in ${\cal X}$ 
from which $Q$ is generated. Then we have
\begin{eqnarray}
R &\leq& 
\frac{p\left(Q - \left(\bigcup_{i=1}^{j-1} \bigcup_{\kappa'=1}^{n(X^A_i)} Z^{[i]}_{\kappa'} \right)\right)}
{c(Q)}  \label{ineq:greedy-Z-1}\\
&=& \frac{p(X - \bigcup_{i=1}^{j-1} X^A_{i})}{c(X)} \label{ineq:greedy-Z-2} \\
&\leq& 
\frac{p\left(Z_{\kappa}^{[j]} - \left(\bigcup_{i=1}^{j-1} \bigcup_{\kappa'=1}^{n(X^A_i)} Z^{[i]}_{\kappa'} \right) 
\right)}{c(Z_{\kappa}^{[j]})} \label{ineq:greedy-Z-3}\\
&=& 
\frac{p\left(Z_{\kappa}^{[j]} - \left(\bigcup_{i=1}^{j-1} \bigcup_{\kappa'=1}^{n(X^A_i)} Z^{[i]}_{\kappa'} \right) 
- \left( \bigcup_{\kappa'=1}^{\kappa-1} Z^{[j]}_{\kappa'} \right)\right)}{c(Z_{\kappa}^{[j]})} 
\label{ineq:greedy-Z-4}
\end{eqnarray}

Inequality (\ref{ineq:greedy-Z-1}) holds since the set whose probability is 
considered at the numerator of the right hand side of (\ref{ineq:greedy-Z-1}) is a superset of the set 
whose probability is 
considered at the numerator of the right hand side of (\ref{eq:Z-greedy}).

Inequality  (\ref{ineq:greedy-Z-2}) follows from (\ref{ineq:greedy-Z-1}) by property (ii) above.

In order to prove (\ref{ineq:greedy-Z-3}) we consider two cases, according to whether $j = r+1$ or $j < r+1.$

If $j < r+1$, the first inequality below follows from the greedy choice  

$$\frac{p(X - \bigcup_{i=1}^{j-1} X^A_{i})}{c(X)} \leq 
\frac{p(X^A_j - \bigcup_{i=1}^{j-1} X^A_{i})}{c(X)} = 
\frac{p\left(Z_{\kappa}^{[j]} - \left(\bigcup_{i=1}^{j-1} \bigcup_{\kappa'=1}^{n(X^A_i)} Z^{[i]}_{\kappa'} \right) 
\right)}{c(Z_{\kappa}^{[j]})} $$
and  the last equality follows from property (ii) of the  proposition under analysis.

If $j = r+1$ we have that, by definition\footnote{Recall that $Y$ is the test in ${\cal X}$ which maximizes the
greedy criterium, but is not chosen because it exceeds the available budget} of $Y$ and the sequence ${\bf Z}^{[r+1]},$ it holds that  
$$\frac{p(X - \bigcup_{i=1}^{j-1} X^A_{i})}{c(X)} \leq 
\frac{p(Y - \bigcup_{i=1}^{j-1} X^A_{i})}{c(X)} = 
\frac{p\left(Z_{\kappa}^{[j]} - \left(\bigcup_{i=1}^{j-1} \bigcup_{\kappa'=1}^{n(X^A_i)} Z^{[i]}_{\kappa'} \right) 
\right)}{c(Z_{\kappa}^{[j]})} $$
where the last equality follows from property (ii) above.

Finally, (\ref{ineq:greedy-Z-4}) follows from (\ref{ineq:greedy-Z-3}) because of property (i) above, from which 
we have that  $Z_{\kappa}^{[j]} \cap \left( \bigcup_{\kappa'=1}^{\kappa-1} Z^{[j]}_{\kappa'} \right) = \emptyset$
hence, \\
$p\left(Z_{\kappa}^{[j]} - \left(\bigcup_{i=1}^{j-1} \bigcup_{\kappa'=1}^{n(X^A_i)} Z^{[i]}_{\kappa'} \right) 
- \left( \bigcup_{\kappa'=1}^{\kappa-1} Z^{[j]}_{\kappa'} \right)\right) = 
p\left(Z_{\kappa}^{[j]} - \left(\bigcup_{i=1}^{j-1} \bigcup_{\kappa'=1}^{n(X^A_i)} Z^{[i]}_{\kappa'} \right)\right).$

\end{proof}

\end{document}